\documentclass[11pt,a4paper]{article}

\usepackage[margin=1in]{geometry}

\usepackage{path}
\usepackage{amssymb,amsthm,amsmath}
\usepackage{thmtools,thm-restate}
\usepackage{booktabs}
\usepackage{enumerate}
\usepackage{graphicx}
\usepackage{multicol}
\graphicspath{{FIG/}}
\newcommand{\fig}[1]{\figurename~\ref{#1}}

\declaretheorem[name=Lemma]{lemma}
\declaretheorem[name=Corollary]{corollary}

\declaretheorem[name=Proposition]{prop}

\newcommand{\floor}[1]{\left \lfloor #1 \right \rfloor}
\newcommand{\ceil}[1]{\left \lceil #1 \right \rceil}

\newcommand{\lsidef}{\operatorname{\mathsf{l}}}
\newcommand{\rsidef}{\operatorname{\mathsf{r}}}
\newcommand{\lside}[1]{\lsidef(#1)}
\newcommand{\rside}[1]{\rsidef(#1)}
\newcommand{\oppositef}{\operatorname{\mathsf{o}}}
\newcommand{\multiplicityf}{\operatorname{\mathsf{m}}}
\newcommand{\opposite}[1]{\oppositef(#1)}
\newcommand{\multiplicity}[1]{\multiplicityf(#1)}

\newcommand{\varnb}[1]{\operatorname{\mathsf{nb}}_{#1}}
\newcommand{\nb}[2]{\operatorname{\mathsf{nb}}_{#1}(#2)}

\newcommand{\embr}[1]{\mathsf{embr}_{#1}}

\newcommand{\veczero}{\mathbf{0}}
\newcommand{\vecone}{\mathbf{1}}

\newcommand{\N}{\mathbb{N}}
\newcommand{\R}{\mathbb{R}}

\newcommand{\Pconv}{P_{\textsf{con}}}
\newcommand{\Psym}{P_{\textsf{sym}}}
\newcommand{\Pbarely}{P_{\textsf{bar}}}
\newcommand{\origin}{\mathbf{0}}

\newcommand{\conv}[1]{{\sf conv}(#1)}

\newcommand{\fillup}[1]{#1^{\cdot\cdot\cdot}}

\newcommand{\commadots}{,\ldots,}

\newcommand{\side}{\sigma}

\usepackage{tikz}

\tikzstyle{Wpt} = [
  thick,
  circle,
  draw=black,
  fill=white,
  inner sep=1.2pt
]
\tikzstyle{Hpt} = [
  circle,
  draw=black,
  fill=black,
  inner sep=1.2pt
]
\tikzstyle{sptline0} = [
  thick,
  dotted
]
\tikzstyle{sptline} = [
  thick,
  dotted,
  shorten >= -15pt,
  shorten <= -20pt
]
\tikzstyle{sptline2} = [
  thick,
  dotted,
  shorten >= -15pt,
  shorten <= -45pt
]
\tikzstyle{sptline3} = [
  thick,
  color=black!50!white
]
\tikzstyle{sptline4} = [
  color=black!50!white,
  thick,
  densely dotted
]

\usetikzlibrary{backgrounds}
\usetikzlibrary{decorations.markings}
\usetikzlibrary{snakes}

\usepackage[blocks]{authblk}

\title{
From Crossing-Free Graphs on Wheel Sets \\ to
Embracing Simplices \\ and Polytopes with Few Vertices%
\footnote{Preliminary version appeared in Proc.\ 33rd International Symposium on Computational Geometry (SoCG 2017), volume 77 of LIPIcs, pages 54:1-54:16. Schloss Dagstuhl - Leibniz-Zentrum fuer Informatik, 2017.}
}
\author[1]{Alexander Pilz\thanks{Supported by a Schr\"odinger fellowship of the Austrian Science Fund (FWF): J-3847-N35.}}
\author[2]{Emo Welzl}
\author[2]{Manuel Wettstein}
\affil[1]{Institute of Software Technology, Graz University of Technology, Austria.
\texttt{apilz@ist.tugraz.at}.}
\affil[2]{Department of Computer Science, ETH Z\"urich, Switzerland.
\texttt{\{emo,mw\}@inf.ethz.ch}.
}
 
\date{August 30, 2019}

\begin{document}
\clearpage\maketitle
\thispagestyle{empty}


\begin{abstract}

A set~$P = H \cup \{w\}$ of~$n+1$ points in general position in the plane is called a \emph{wheel set} if all points but~$w$ are extreme.
We show that for the purpose of counting crossing-free geometric graphs on such a set~$P$, it suffices to know the \emph{frequency vector} of~$P$.
While there are roughly~$2^n$ distinct order types that correspond to wheel sets, the number of frequency vectors is only about~$2^{n/2}$.

We give simple formulas in terms of the frequency vector for the number of crossing-free spanning cycles, matchings, triangulations, and many more.
Based on that, the corresponding numbers of graphs can be computed efficiently.
In particular, we rediscover an already known formula for~$w$-embracing triangles spanned by~$H$.

Also in higher dimensions, wheel sets turn out to be a suitable model to approach the problem of computing the \emph{simplicial depth} of a point~$w$ in a set~$H$, i.e., the number of~$w$-embracing simplices.
While our previous arguments in the plane do not generalize easily, we show how to use similar ideas in~$\R^d$ for any fixed~$d$.
The result is an~$O(n^{d-1})$ time algorithm for computing the simplicial depth of a point~$w$ in a set~$H$ of~$n$ points, improving on the previously best bound of~$O(n^d\log n)$.

Based on our result about simplicial depth, we can compute the number of facets of the convex hull of $n=d+k$ points in general position in~$\R^d$ in time~$O(n^{\max\{\omega,k-2\}})$ where~$\omega \approx 2.373$, even though the asymptotic number of facets may be as large as~$n^k$.
\end{abstract}

\newpage
\setcounter{page}{1}


\section{Introduction}
Computing the number of non-crossing straight-line drawings of certain graph classes (triangulations, spanning trees, etc.) on a planar point set is a well-known problem in computational and discrete geometry.
While for point sets in convex position many of these numbers have simple closed formulas, it seems difficult to compute them efficiently for a given general point set, or to provide tight upper and lower bounds.
In this paper, we provide means for solving these problems for a special class of point sets that we call wheel sets.

\subparagraph*{Conowheel sets.}
Let $P = H \cup \{w\}$ be a set of $n+1$ points in the plane.
Unless stated otherwise, $P$ is assumed to be in general position (i.e., no three points on a common line) and the points in $H$ are assumed to be extreme (i.e., vertices of the convex hull of $P$).
$P$ is in \emph{convex position} if all points including $w$ are extreme, and $P$ is a \emph{wheel set} if all points except $w$ are extreme.
If $P$ is either of them, then we call it a \emph{conowheel set}.

There are three special conowheel sets that reappear throughout the paper in proofs and examples; let us introduce notation in order to be able to easily refer to them later on.
First, we denote by $\Pconv$ a concrete set in convex position (say, the vertex set of a regular $(n+1)$-gon).
Second, we denote by $\Pbarely$ a \emph{barely-in wheel set} (i.e., $H$ is the vertex set of a regular $n$-gon and $w$ is sufficiently close to an edge $e$ of the $n$-gon in such a way that $w$ is in the interior of every triangle spanned by $e$ and a third point of $H$).
Third, we denote by $\Psym$ a \emph{symmetric wheel set} (i.e., $H$ is again the vertex set of a regular $n$-gon with $w$ at its center, perturbed slightly so as to obtain general position).
Example drawings of all three sets for $n=7$ can be seen in Figure~\ref{fig:examples}.

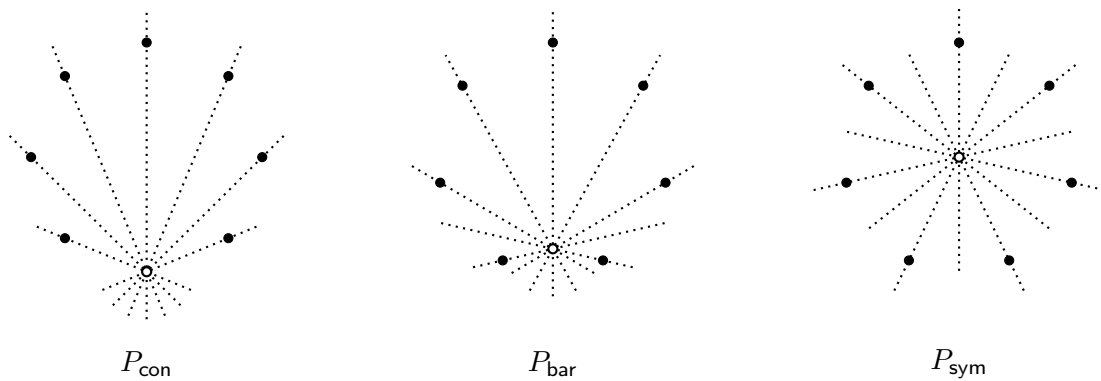
\begin{figure}[b]
  \begin{center}
\begin{tikzpicture}[scale=1.52]
  \begin{scope}[xshift=0]
    \node[Wpt] (w) at (270:1) {};
    \node[Hpt] (h1) at (90+0*360/8:1) {};
    \node[Hpt] (h2) at (90+1*360/8:1) {};
    \node[Hpt] (h3) at (90+2*360/8:1) {};
    \node[Hpt] (h4) at (90+3*360/8:1) {};
    \node[Hpt] (h5) at (90+5*360/8:1) {};
    \node[Hpt] (h6) at (90+6*360/8:1) {};
    \node[Hpt] (h7) at (90+7*360/8:1) {};
    \begin{scope}[on background layer]
      \draw[sptline] (w) -- (h1);
      \draw[sptline] (w) -- (h2);
      \draw[sptline] (w) -- (h3);
      \draw[sptline] (w) -- (h4);
      \draw[sptline] (w) -- (h5);
      \draw[sptline] (w) -- (h6);
      \draw[sptline] (w) -- (h7);
    \end{scope}
    \node at (0,-1.8) {$\Pconv$};
  \end{scope}
  \begin{scope}[xshift=100]
    \node[Wpt] (w) at (270:0.8) {};
    \node[Hpt] (h1) at (90+0*360/7:1) {};
    \node[Hpt] (h2) at (90+1*360/7:1) {};
    \node[Hpt] (h3) at (90+2*360/7:1) {};
    \node[Hpt] (h4) at (90+3*360/7:1) {};
    \node[Hpt] (h5) at (90+4*360/7:1) {};
    \node[Hpt] (h6) at (90+5*360/7:1) {};
    \node[Hpt] (h7) at (90+6*360/7:1) {};
    \begin{scope}[on background layer]
      \draw[sptline] (w) -- (h1);
      \draw[sptline] (w) -- (h2);
      \draw[sptline] (w) -- (h3);
      \draw[sptline2] (w) -- (h4);
      \draw[sptline2] (w) -- (h5);
      \draw[sptline] (w) -- (h6);
      \draw[sptline] (w) -- (h7);
    \end{scope}
    \node at (0,-1.8) {$\Pbarely$};
  \end{scope}
  \begin{scope}[xshift=200]
    \node[Wpt] (w) at (0,0) {};
    \node[Hpt] (h1) at (90+0*360/7:1) {};
    \node[Hpt] (h2) at (90+1*360/7:1) {};
    \node[Hpt] (h3) at (90+2*360/7:1) {};
    \node[Hpt] (h4) at (90+3*360/7:1) {};
    \node[Hpt] (h5) at (90+4*360/7:1) {};
    \node[Hpt] (h6) at (90+5*360/7:1) {};
    \node[Hpt] (h7) at (90+6*360/7:1) {};
    \begin{scope}[on background layer]
      \draw[sptline2] (w) -- (h1);
      \draw[sptline2] (w) -- (h2);
      \draw[sptline2] (w) -- (h3);
      \draw[sptline2] (w) -- (h4);
      \draw[sptline2] (w) -- (h5);
      \draw[sptline2] (w) -- (h6);
      \draw[sptline2] (w) -- (h7);
    \end{scope}
    \node at (0,-1.8) {$\Psym$};
  \end{scope}
\end{tikzpicture}
  \end{center}
  \caption{Convex position, a barely-in, and a symmetric wheel set. The extra point $w$ is drawn in white. The dotted supporting lines are instrumental in the definition of the frequency vector.}
  \label{fig:examples}
\end{figure}

The numbers of triangulations, pseudo-triangulations, and perfect matchings on wheel sets have been studied~\cite{Randall01, RuizW15}.
Our work generalizes these approaches.
Wheel sets have also been useful in the investigation of high-dimensional polytopes with few vertices; already in the 1960s, Perles has counted so-called distended standard forms of Gale diagrams of simplicial $d$-polytopes with at most $d+3$ vertices (as reported by Gr\"unbaum~\cite[Chapter~6.3]{gruenbaum_book}).
In the terminology of modern discrete geometry, these correspond to the different order types of wheel sets.

\subparagraph*{Order types.}
The \emph{order type} of a point set $P$ in general position is a combinatorial description that assigns an orientation (either clockwise or counterclockwise) to every ordered triple of points~\cite{multidimensional_sorting}.
Two point sets are then said to have the same order type if there exists a bijection between the two sets that preserves these orientations.
In this paper, however, we follow the practice of also considering two point sets to have the same order type if there exists a bijection that reverses all orientations (i.e., reflected point sets have the same order type).

Many combinatorial properties of a point set can be recovered from its order type.
In particular, the order type determines whether two segments with endpoints in $P$ cross, and whether a given point in $P$ is extreme.
It is not hard to see that all sets in convex position have the same order type.
However, the same is not true for wheel sets.

\begin{restatable*}{theorem}{thmordtype}
  \label{thm:ordtype}
  For any $n \geq 1$, the number of order types of conowheel sets of size $n+1$ is\footnote{Here, $\varphi(k)$ denotes Euler's totient function, which counts the integers coprime to $k$ that are at most $k$.}
  \begin{equation*}
    \frac{1}{4n}\sum_{2 \nmid k|n} \varphi(k)2^{n/k}+2^{\floor{(n-3)/2}} = \Theta(2^n/n) \enspace .
  \end{equation*}
\end{restatable*}

The above formula also counts so-called self-dual 2-colored necklaces~\cite{Brouwer80,PalmerR84}.\footnote{Sequence A007147 on OEIS (Online Encyclopedia of Integer Sequences).}
In Section~\ref{sec:ordtypefreqvec} we give a proof of the theorem by describing a bijection between order types of conowheel sets and such necklaces.
The formula itself, however, has first been obtained by Perles (as stated, without proof, in~\cite[Chapter~6.3]{gruenbaum_book}) in his investigation into the number of different combinatorial types of simplicial $d$-polytopes with at most $n=d+3$ vertices. 
The correspondence between simplicial polytopes and self-dual necklaces has in turn been estalished by Montellano-Ballesteros and Strausz~\cite{MontellanoS04} using Radon complexes.

\subparagraph*{Frequency vectors.}
While the order type of a point set determines the set of non-crossing straight-line graphs on it, we will see that the number of non-crossing graphs on a wheel set is already determined by the so-called frequency vector, which is defined as follows.

Let $P = H \cup \{w\}$ be a conowheel set and let $h \in H$ be arbitrary.
Let $\lside{h}$ denote the number of points strictly to the left of the directed line going from $w$ to $h$, and let $\rside{h}$ denote the number of points strictly to the right of that line.
The \emph{frequency vector} of $P$ is the vector $F(P) = (F_0,F_1,\dots,F_{n-1})$ where $F_i$ is the number of points $h \in H$ satisfying $|\lside{h}-\rside{h}| = i$.

One can easily verify the following examples for $n=7$ in Figure~\ref{fig:examples}.
\begin{align*}
  F(\Pconv) = (1,0,2,0,2,0,2) &&
  F(\Pbarely) = (1,0,2,0,4,0,0) &&
  F(\Psym) = (7,0,0,0,0,0,0)
\end{align*}

Note that the frequency vector can be computed in $O(n \log n)$ time by radially sorting~$H$ around~$w$.
It is also clear that the order type determines the frequency vector.
However, the converse is not true.
In Section~\ref{sec:ordtypefreqvec} we give a characterization of frequency vectors, which allows the following conclusion.

\begin{restatable*}{theorem}{thmfreqvec}
  \label{thm:freqvec}
  For any $n \geq 1$, the number of frequency vectors realizable by a conowheel set of size $n+1$ is $2^{\ceil{n/2}-1}$.
\end{restatable*}

Given that the number of frequency vectors is significantly smaller than the number of order types, it is unclear how much the frequency vector reveals about a conowheel set.
However, we will see that for the purpose of counting non-crossing structures it is both sufficient and necessary.

There is again a connection to polytopes with few vertices.
We will see that frequency vectors of wheel sets of size $n+1$ correspond to $f$-vectors of simplicial $d$-polytopes with at most $n=d+3$ vertices.
Linusson has counted the latter using a sophisticated analysis of so-called $M$-sequences~\cite{Linusson99} and asks for a more direct proof, which will be provided in this paper.

\subparagraph*{Geometric graphs.}
A \emph{geometric graph} on $P$ is a graph with vertex set $P$ and edges drawn as straight segments between the corresponding endpoints, and it is \emph{crossing-free} if no two edges intersect in their respective relative interiors.
Many families of crossing-free geometric graphs have been defined and studied, such as \emph{triangulations}, \emph{perfect matchings}, \emph{spanning trees}, and so on.

There exists a vast literature that is concerned with counting these crossing-free structures on specific point sets or proving extremal upper and lower bounds \cite{AichholzerHHHKV07,SharirS11,SharirSW13,SharirW06}.
One comparatively simple case is if $P$ is in convex position.
In that case, counting triangulations is a classic problem that goes back to Euler, and it gives rise to the famous Catalan numbers.
For many other families of graphs (such as perfect matchings and spanning trees), simple closed formulas can be obtained as well \cite{DulucqP93,FlajoletN99,Motzkin48}.

Randall et al.~\cite{Randall01} have been the first to consider geometric graphs on wheel sets.
They have found the extremal configurations for triangulations and pseudo-triangulations by using an argument that involves continuously moving around the extra point $w$.
We will follow a similar approach in this paper.
The case of perfect matchings has been studied by Ruiz-Vargas and Welzl~\cite{RuizW15}.
The next theorem, as proved in Section~\ref{thm:ggraph}, is a generalization of a result from their paper.

In the following, let $\mathcal{G}$ be a set of abstract (unlabeled) graphs with $n+1$ vertices, and let $\nb{\mathcal{G}}{P}$ denote the number of crossing-free geometric graphs on $P$ that are isomorphic to a graph in $\mathcal{G}$.
In other words, $\nb{\mathcal{G}}{P}$ is the number of non-crossing straight-line embeddings of graphs in $\mathcal{G}$ on~$P$.

\begin{restatable*}{theorem}{thmggraph}
  \label{thm:ggraph}
  Let $\mathcal{G}$ be arbitrary, and let $P = H \cup \{w\}$ be a conowheel set of size $n+1$.
  Then, $\nb{\mathcal{G}}{P}$ depends only on the frequency vector $F(P) = (F_0,F_1,\dots,F_{n-1})$.
  More concretely,
  \begin{equation*}
    \nb{\mathcal{G}}{P} = \gamma_n - \frac{1}{2} \sum_{h \in H} \lambda_{\lside{h},\rside{h}} = \sum_{k=0}^{n-1} F_k \Lambda_k \enspace ,
  \end{equation*}
  where $\gamma_n$ and $\lambda_{l,r} = \lambda_{r,l}$ are integers and $\Lambda_k$ are rationals depending on $\mathcal{G}$.
\end{restatable*}

While the latter formula in the above theorem makes the dependency on the frequency vector more obvious, the former will turn out to be more natural.
The latter formula follows from the former simply by putting $\Lambda_k = \gamma_n/n + 1/2 \cdot \lambda_{(n+k-1)/2, (n-k-1)/2}$.

In Section~\ref{sec:originembracing} we further show that the converse of Theorem~\ref{thm:ggraph} also holds.
That is, the frequency vector is determined by the number of non-crossing embeddings of certain graphs.

\begin{restatable*}{theorem}{thmggraphconverse}
  \label{thm:ggraphconverse}
  Let $P$ and $P'$ be two conowheel sets of the same size.
Then, $\nb{\mathcal{G}}{P} = \nb{\mathcal{G}}{P'}$ holds for every graph class $\mathcal{G}$ if and only if $F(P) = F(P')$.
\end{restatable*}

We give just one example here, which at the same time makes the connection to the later parts of the paper.
Let $\mathcal{G} = \{\fillup{K_4}\}$, where $\fillup{K_4}$ is obtained by adding $n-3$ additional isolated vertices to the complete graph $K_4$.
The following formula is obtained alongside the proof of Theorem~\ref{thm:ggraph}:
\begin{equation}
  \label{eq:embtrg}
  \nb{\mathcal{G}}{P}
  = \binom{n}{3} - \frac{1}{2} \sum_{h \in H} \left(\binom{\lside{h}}{2} + \binom{\rside{h}}{2}\right)
  = \binom{n}{3} - \sum_{h \in H} \binom{\lside{h}}{2} \hspace{0.5cm} \text{for $\mathcal{G} = \{\fillup{K_4}\}$} \enspace.
\end{equation}

Observe that all non-crossing embeddings of $\fillup{K_4}$ on a given conowheel set $P = H \cup \{w\}$ have the following property:
One of the vertices of the underlying $K_4$ is mapped to the point $w$, while the other three vertices are mapped to three points that form a triangle that contains the extra point $w$ in its interior.
With equation (\ref{eq:embtrg}) we thus get a rather simple formula for the number of \emph{$w$-embracing triangles} (i.e., point triples in $H$ whose convex hull contains~$w$).

We also note that the second (and simpler) formula in equation~(\ref{eq:embtrg}) can be derived with a more direct argument, essentially by subtracting all non-$w$-embracing triangles from the set of all triangles spanned by $H$. 
This observation has been made already earlier in~\cite{rousseeuw}.

Further note that the set of $w$-embracing triangles does not change if we replace a point $h \in H$ by any other point $\tilde h$ on the ray that starts at $w$ and passes through $h$.
For counting $w$-embracing triangles, the approach for conowheel sets thus generalizes to arbitrary point sets.

\subparagraph*{Higher dimensions.}
The concept of conowheel sets can be generalized to arbitrary dimensions, where we may again consider sets with at most one non-extreme point.
However, even for counting $w$-embracing tetrahedra in $3$-space, the ideas from the proof of Theorem~\ref{thm:ggraph} do not generalize easily.
Nevertheless, in Section \ref{sec:higher_dimensions} we give a generalization of equation~(\ref{eq:embtrg}).
From that we obtain improved time bounds for computing the number of $w$-embracing simplices or, in other words, the \emph{simplicial depth} of a point $w$ in $H$, as defined in~\cite{liu90}.

\begin{restatable*}{theorem}{thmsimpldepth}
  \label{thm:simpldepth}
  Let $d \geq 3$ be fixed and let $H$ be a set of $n$ points in $\R^d$.
  Then, the simplicial depth of a point $w$ in $H$ can be computed in $O(n^{d-1})$ time.
\end{restatable*}

Again, this result is stated for arbitrary sets $H$ and not for wheel sets only, as for the simplicial depth only the position relative to $w$ is relevant.
We further note that the algorithm generalizes to counting all $k$-element subsets of $H$ whose convex hull contains~$w$.

The simplicial depth of a point has attracted considerable attention as a measure of data depth.
Several authors have investigated the problem in the plane~\cite{geometric_medians,khuller_mitchell,rousseeuw}.
Algorithms that run in time $O(n^2)$ and $O(n^4)$ in $\R^3$ and $\R^4$, respectively, are provided by Cheng and Ouyang~\cite{d_4}, who also point out flaws in previous algorithms in 3-space.
Our result improves over the best known general $O(n^d\log n)$ time algorithm for points in constant dimension~$d$~\cite{afshani}.
For $d$ not constant, the problem is known to be $\#P$-complete and $W[1]$-hard~\cite{afshani}.

One of the aims of the work by Perles has been to count facets of high-dimensional simplicial polytopes with few vertices.
Via the so-called Gale dual, this number is equal to the number of simplices embracing the origin in a low-dimensional dual point set.
In Section~\ref{sec:higher_dimensions} we further show how to apply Theorem~\ref{thm:simpldepth} in order to compute the number of facets of such a polytope.

\begin{restatable*}{theorem}{thmcalcfacet}
\label{thm:calc_facet}
Let $k \geq 1$ be fixed and let $S$ be a set of $n=d+k$ points in general position in $\R^{d}$.
Then, the number of facets of the simplicial polytope $\conv{S}$ can be computed in $O(n^{k-2})$ time if $k \geq 5$ and in $O(n^\omega)$ time otherwise, where $\omega \approx 2.373$. 
\end{restatable*}


\section{Order Types and Frequency Vectors}
\label{sec:ordtypefreqvec}

The purpose of this section is to give an explanation for Table~\ref{tbl:numbers}.
The latter contains the numbers of distinct order types and frequency vectors corresponding to conowheel sets of size $n+1$.
For completeness, we have also included the corresponding numbers if equivalence over order types is defined to not include reflections.

\begin{table}[b]
  \centering
  {\small
  \begin{tabular}{rrrrcrrrr}
    \toprule
       & \multicolumn{2}{c}{\bf Order Types} & \multicolumn{1}{c}{\bf Freq.\ Vectors} & &
       & \multicolumn{2}{c}{\bf Order Types} & \multicolumn{1}{c}{\bf Freq.\ Vectors} \\ 
   $n$ & with & w/o                          & & &
   $n$ & with & w/o                          & \\
       & \multicolumn{2}{c}{reflection}      & & &
       & \multicolumn{2}{c}{reflection}      & \\
    \cmidrule{1-4}
    \cmidrule{6-9}
      1 & 1 & 1 & 1 & & 7 & 9 & 10 & 8 \\
      2 & 1 & 1 & 1 & & 8 & 12 & 16 & 8 \\
      3 & 2 & 2 & 2 & & 9 & 23 & 30 & 16 \\
      4 & 2 & 2 & 2 & & 10 & 34 & 52 & 16 \\
      5 & 4 & 4 & 4 & & 11 & 63 & 94 & 32 \\
      6 & 5 & 6 & 4 & & 12 & 102 & 172 & 32 \\
    \bottomrule
  \end{tabular}
  }
  \caption{Number of order types and frequency vectors of conowheel sets over $n+1$ points.}
  \label{tbl:numbers}
\end{table}


\subparagraph*{Order types.}
Given a set $H$ of $n=7$ points forming the vertex set of a regular heptagon, as on the left hand side of Figure~\ref{fig:heptagon}, there are eight conowheel sets $P = H \cup \{w\}$ with distinct order types that can be obtained by adding an extra point $w$.
By first deforming $H$, as illustrated on the right hand side of Figure~\ref{fig:heptagon}, one obtains the ninth and last order type for $n=7$, giving the count listed in Table~\ref{tbl:numbers}.
This necessary deformation of $H$ seems to complicate matters significantly, but only at first sight.

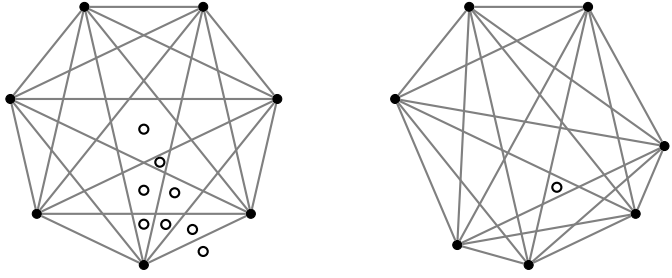
\begin{figure}
  \centering
  \begin{tikzpicture}[rotate=270,scale=1.8]
    \begin{scope}
       \node[Hpt] (p1) at (1*360/7:1) {};
       \node[Hpt] (p2) at (2*360/7:1) {};
       \node[Hpt] (p3) at (3*360/7:1) {};
       \node[Hpt] (p4) at (4*360/7:1) {};
       \node[Hpt] (p5) at (5*360/7:1) {};
       \node[Hpt] (p6) at (6*360/7:1) {};
       \node[Hpt] (p7) at (7*360/7:1) {};
       \node[Wpt] at (0,0) {};
       \node[Wpt] at (0:0.45) {};
       \node[Wpt] at (0:0.70) {};
       \node[Wpt] at (360/14:0.27) {};
       \node[Wpt] at (360/14:0.52) {};
       \node[Wpt] at (360/14:0.82) {};
       \node[Wpt] at (360/14:1.0) {};
       \node[Wpt] at (360/14-360/28:0.72) {};
       \draw[sptline3] (p1) -- (p2);
       \draw[sptline3] (p1) -- (p3);
       \draw[sptline3] (p1) -- (p4);
       \draw[sptline3] (p1) -- (p5);
       \draw[sptline3] (p1) -- (p6);
       \draw[sptline3] (p1) -- (p7);
       \draw[sptline3] (p2) -- (p3);
       \draw[sptline3] (p2) -- (p4);
       \draw[sptline3] (p2) -- (p5);
       \draw[sptline3] (p2) -- (p6);
       \draw[sptline3] (p2) -- (p7);
       \draw[sptline3] (p3) -- (p4);
       \draw[sptline3] (p3) -- (p5);
       \draw[sptline3] (p3) -- (p6);
       \draw[sptline3] (p3) -- (p7);
       \draw[sptline3] (p4) -- (p5);
       \draw[sptline3] (p4) -- (p6);
       \draw[sptline3] (p4) -- (p7);
       \draw[sptline3] (p5) -- (p6);
       \draw[sptline3] (p5) -- (p7);
       \draw[sptline3] (p6) -- (p7);
    \end{scope}
       
     \begin{scope}[yshift=80]
       \node[Hpt] (p1) at (1*360/7:1) {};
       \node[Hpt] (p2) at (2*360/7-20:1) {};
       \node[Hpt] (p3) at (3*360/7:1) {};
       \node[Hpt] (p4) at (4*360/7:1) {};
       \node[Hpt] (p5) at (5*360/7:1) {};
       \node[Hpt] (p6) at (6*360/7+20:1) {};
       \node[Hpt] (p7) at (7*360/7:1) {};
       \node[Wpt] (w) at (360/14:0.475) {};
       \draw[sptline3] (p1) -- (p2);
       \draw[sptline3] (p1) -- (p3);
       \draw[sptline3] (p1) -- (p4);
       \draw[sptline3] (p1) -- (p5);
       \draw[sptline3] (p1) -- (p6);
       \draw[sptline3] (p1) -- (p7);
       \draw[sptline3] (p2) -- (p3);
       \draw[sptline3] (p2) -- (p4);
       \draw[sptline3] (p2) -- (p5);
       \draw[sptline3] (p2) -- (p6);
       \draw[sptline3] (p2) -- (p7);
       \draw[sptline3] (p3) -- (p4);
       \draw[sptline3] (p3) -- (p5);
       \draw[sptline3] (p3) -- (p6);
       \draw[sptline3] (p3) -- (p7);
       \draw[sptline3] (p4) -- (p5);
       \draw[sptline3] (p4) -- (p6);
       \draw[sptline3] (p4) -- (p7);
       \draw[sptline3] (p5) -- (p6);
       \draw[sptline3] (p5) -- (p7);
       \draw[sptline3] (p6) -- (p7);
     \end{scope}
    \end{tikzpicture}
  \caption{All order types of conowheel sets for $n=7$. Each white point $w$ combined with the surrounding set $H$ of black points yields one distinct set.}
  \label{fig:heptagon}
\end{figure}

As already noted in the introduction, the formula in Theorem~\ref{thm:ordtype} has also been obtained in the context of counting  \emph{self-dual 2-colored necklaces} with $2n$ beads and with mirrored necklaces identified~\cite{Brouwer80,PalmerR84}.
These are binary (say, black and white) circular sequences of length $2n$ such that elements at distance $n$ (i.e., opposing beads) are distinct (i.e., if one is black the other must be white, and vice versa).
We give a proof for Theorem~\ref{thm:ordtype} by using a simple bijection to such necklaces.

We also note that a similar (and slightly simpler) formula is known if mirrored necklaces are not identified.\footnote{Sequence A000016 on OEIS.}
Naturally, such a formula also counts order types of conowheel sets without reflection.

{\renewcommand{\footnote}[1]{} 
\thmordtype}
\begin{proof}
  Let $P = H \cup \{w\}$ be a conowheel set.
  Consider a directed line $s$ containing $w$ that rotates counterclockwise with $w$ as a hub by $2\pi$.
  The line passes over each point in $H$ twice, once on the positive ray and once on the negative ray.
  We record the sequence in which the points $h \in H$ are passed, and indicate for each entry whether the corresponding point $h$ was on the positive or negative ray of $s$.
  This sequence can be considered cyclic, and is known as the \emph{local sequence} of $w$ \cite{semispaces}.
  It depends only on the order type of $P$, and it naturally corresponds to a self-dual necklace with $2n$ beads and two colors (say, positive becomes black and negative becomes white), as illustrated in Figure~\ref{fig:bijection}.
  
  \begin{figure}[b]
    \begin{center}
      \begin{tikzpicture}
           \begin{scope}[xshift=0]
               \node[Wpt] (w) at (0,0) {};
               \node[Hpt] (h1) at (30:1) {};
               \node[Hpt] (h2) at (50:0.8) {};
               \node[Hpt] (h3) at (90:-0.8) {};
               \node[Hpt] (h4) at (150:-1) {};
               \node[Hpt] (h5) at (170:1.3) {};
               \begin{scope}[on background layer]
                 \draw[->] (98:1.3)  arc (98:130:1.3);
                 \draw[sptline2] (w) -- (h1);
                 \draw[sptline2] (w) -- (h2);
                 \draw[sptline2] (w) node[right=7pt,above=30pt] {$s$} -- (h3);
                 \draw[sptline2] (w) -- (h4);  
                 \draw[sptline2] (w) -- (h5);
               \end{scope}
           \end{scope}
           \begin{scope}[xshift=60]
             \node {$\mapsto$};
           \end{scope}
           \begin{scope}[xshift=100,scale=0.75,rotate=360/20]
             \tikzstyle{beada} = [
               circle,
               draw = black,
               fill = white,
               inner sep = 0pt,
               minimum width = 5.5pt
             ]
             \tikzstyle{beadb} = [
               circle,
               draw = black,             
               fill = black,
               inner sep = 0pt,
               minimum width = 5.5pt
             ]
               \draw (0,0) circle (1);
                 \node[beada] (b0) at (-1*36:1) {};
                 \node[beadb] (b1) at (0*36:1) {};
                 \node[beadb] (b2) at (1*36:1) {};
                 \node[beada] (b3) at (2*36:1) {};
                 \node[beada] (b4) at (3*36:1) {};
                 \node[beadb] (b5) at (4*36:1) {};
                 \node[beada] (b6) at (5*36:1) {};
                 \node[beada] (b7) at (6*36:1) {};
                 \node[beadb] (b8) at (7*36:1) {};
                 \node[beadb] (b9) at (8*36:1) {};
           \end{scope}
           \begin{scope}[xshift=140]
             \node {$\mapsto$};
           \end{scope}
           \begin{scope}[xshift=190,scale=0.75,rotate=360/20]
               \draw[sptline0] (0*36:1.5) -- (0*36:-1.5);
               \draw[sptline0] (1*36:1.5) -- (1*36:-1.5);
               \draw[sptline0] (2*36:1.5) -- (2*36:-1.5);
               \draw[sptline0] (3*36:1.5) -- (3*36:-1.5);
               \draw[sptline0] (4*36:1.5) -- (4*36:-1.5);
               \node[Wpt] at (0,0) {};
               \draw[sptline3] (0,0) circle (1);
               \node[Hpt] at (-1*36:-1) {};
               \node[Hpt] at (0*36:1) {};
               \node[Hpt] at (1*36:1) {};
               \node[Hpt] at (2*36:-1) {};
               \node[Hpt] at (3*36:-1) {};
           \end{scope}
      \end{tikzpicture}
    \end{center}
    \caption{Bijection between order types of wheel sets and self-dual 2-colored necklaces.}
    \label{fig:bijection}
  \end{figure}
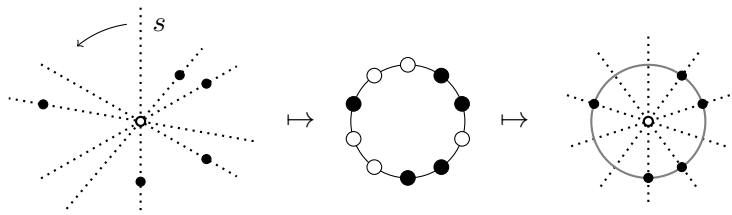
  
  The above mapping is seen to be a bijection by considering its inverse.
  Given a necklace, we can transform it into an order type by placing $w$ at the center of a circle with $2n$ equally spaced vertices on its boundary, by identifying the beads of the necklace with those vertices in circular order, and then by placing a point $h$ on each vertex that corresponds to a black bead.
  By construction, the resulting point set $P = H \cup \{w\}$ is a conowheel set with the same order type that we started with.
  
  The asymptotic estimate is explained by taking the dominant summand with $k=1$.
\end{proof} 


\subparagraph*{Frequency vectors.}
The following lemma characterizes frequency vectors.
The proof is again by letting a line rotate about the extra point $w$, and by observing how it dissects the point set $H$ during the process.
More details can be found in \cite{RuizW15}.

\begin{restatable}{lemma}{lemfreqvecchar}
  \label{lem:freqvecchar}
  $F = (F_0,F_1,\dots,F_{n-1}) \in \N^{n}$ is the frequency vector of a conowheel set $P = H \cup \{w\}$ of size $n+1$, i.e., $F = F(P)$, if and only if
  \begin{multicols}{2}
  \begin{enumerate}[(i)]
    \item $\sum_{k=0}^{n-1} F_k = n$,
    \item $F_k = 0$ for all $k \equiv_2 n$,
    \item $F_k$ is even for all $k \geq 1$, and
    \item if $F_k \neq 0$ and $k \geq 2$, then $F_{k-2} \neq 0$.
  \end{enumerate}
  \end{multicols}
\end{restatable}
\begin{proof}
  Let us first prove the ``only if''.
  From the definition of frequency vectors, (i) and (ii) are immediate; (i) just requires there to be $n$ points in $H$, and (ii) simply says that the parity of the difference $|\lside{h}-\rside{h}|$ can never be the same as $n$ given that $\lside{h}+\rside{h}+1=n$.
  
  As for (iii) and (iv), we make use of a particular sequence $s_{1/2},s_{1},\dots,s_n,s_{n+1/2}$ of lines, as illustrated in Figure~\ref{fig:rotation} and as defined in the following.
  Let $s_1,s_2,\dots,s_n$ be the lines that pass through $w$ and one of the $n$ points of $H$, sorted in radial counterclockwise order around $w$.
  Let $s_{i+1/2}$ be any line that is in between $s_i$ and $s_{i+1}$ (indices are understood modulo $n$).
  More precisely, $s_{i+1/2}$ may be any of the intermediate lines that are encountered when transforming $s_i$ into $s_{i+1}$ by a counterclockwise rotation about $w$.
  Finally, give all lines a direction by orienting $s_{1/2}$ arbitrarily and then rotating counterclockwise about $w$ by an angle $\pi$.
  In particular, this means that $s_{1/2}$ and $s_{n+1/2}$ are the same lines with reversed directions.
  \begin{figure}
    \centering
    \begin{tikzpicture}
           \begin{scope}[rotate=360/20,scale=1.223]
               \draw[sptline0] (0*36:1.5) -- (0*36:-1.5) node[fill=white] {$s_1$};
               \draw[sptline0] (1*36:1.5) -- (1*36:-1.5) node[fill=white] {$s_2$};
               \draw[sptline0] (2*36:1.5) -- (2*36:-1.5) node[fill=white] {$s_3$};
               \draw[sptline0] (3*36:1.5) -- (3*36:-1.5) node[fill=white] {$s_4$};
               \draw[sptline0] (4*36:1.5) -- (4*36:-1.5) node[fill=white] {$s_5$};
               \draw[sptline3] (0,0) circle (1);
               \draw[sptline0] (0.5*36:1.5) -- (0.5*36:-1.5) node[left=3pt] {$s_{1+{1 \over 2}}$};
               \draw[sptline0] (1.5*36:1.5) -- (1.5*36:-1.5) node[below=4pt] {$s_{2+{1 \over 2}}$};
               \draw[sptline0] (2.5*36:1.5) -- (2.5*36:-1.5) node[below=4pt] {$s_{3+{1 \over 2}}$};
               \draw[sptline0] (3.5*36:1.5) -- (3.5*36:-1.5) node[right=3pt] {$s_{4+{1 \over 2}}$};
               \draw[sptline0] (4.5*36:1.5) node[left=5pt] {$s_{{1 \over 2}}$} -- (4.5*36:-1.5) node[right=5pt] {$s_{5+{1 \over 2}}$};
               \node[Hpt] at (-1*36:-1) {};
               \node[Hpt] at (0*36:1) {};
               \node[Hpt] at (1*36:1) {};
               \node[Hpt] at (2*36:-1) {};
               \node[Hpt] at (3*36:-1) {};
               \node[Wpt] at (0,0) {};
           \end{scope}
           \begin{scope}[xshift=150,yshift=-15,yscale=0.5]
             \draw [->](0,-3.5)--(0,3.5) node[above=2pt] {$y = g(s_x)$};
             \draw [->](-0.9,0)--(6.0,0) node[right=2pt] {$x$};
             \foreach \x/\xtext in {
               1.5/{1+{1 \over 2}}, 2.5/{2+{1 \over 2}},
               3.5/{3+{1 \over 2}}, 5.5/{5+{1 \over 2}} }
             {\draw (\x cm,2pt ) -- (\x cm,-2pt ) node[anchor=north] {$\scriptscriptstyle\xtext$};}
             \foreach \x/\xtext in {
               0.5/{{1 \over 2}}, 4.5/{4+{1 \over 2}} }
             {\draw (\x cm,2pt ) -- (\x cm,-2pt ) node[anchor=south] {$\scriptscriptstyle\xtext$};}
             \foreach \y/\ytext in {
               -3/-3, -1/-1,
               +1/+1, +3/+3}
             {\draw (1pt,\y cm) -- (-1pt ,\y cm) node[anchor=east] {$\scriptstyle\ytext$};}
             \tikzstyle{marker} = [
               rectangle,
               draw=black,
               fill=white,
               inner sep=2pt
             ];
             \draw[thick] (0.5,-1) node[marker] {}
                -- (1.5,+1) node[marker] {}
                -- (2.5,+3) node[marker] {}
                -- (3.5,+1) node[marker] {}
                -- (4.5,-1) node[marker] {}
                -- (5.5,+1) node[marker] {};
           \end{scope}
    \end{tikzpicture}
    \caption{On the left, the construction of the lines $s_{1/2},s_{1},\dots,s_{n+1/2}$ as in the proof of Lemma \ref{lem:freqvecchar}. The lines are understood directed towards their corresponding label. On the right, a plot of the corresponding sequence $\gamma = g(s_{{1 \over 2}}),\dots,g(s_{n+{1 \over 2}})$.}
    \label{fig:rotation}
  \end{figure}
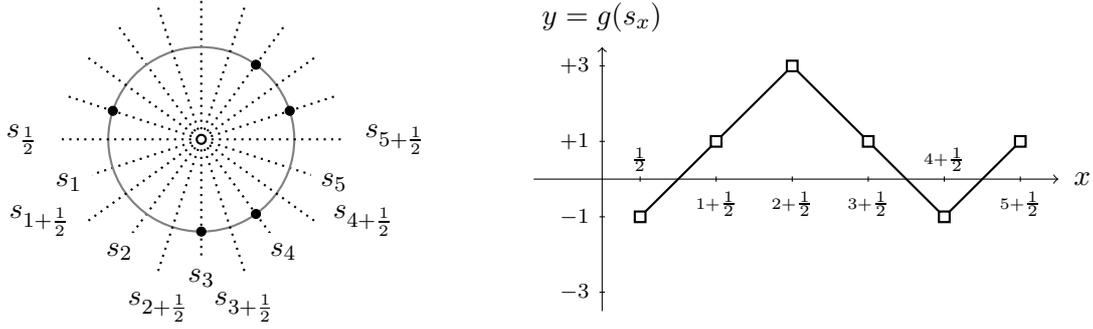
  
  Now, let $\lside{s}$ and $\rside{s}$ denote the number of points of $H$ strictly to the left and right, respectively, of a directed line $s$, and let $g(s) := \lside{s} - \rside{s}$.
  We consider the sequence $\gamma = g(s_{1/2}),g(s_{1+1/2}),\dots,g(s_{n+1/2})$ and make the following three crucial observations.
  First, for any integer~$i$, $g(s_i)$ is the average of $g(s_{i-1/2})$ and $g(s_{i+1/2})$.
  Second, the sequence $\gamma$ is ``continuous'' in the sense that any two subsequent elements differ by exactly $2$ (henceforth called a \emph{jump} over the integer in between).
  Third, the sequence $\gamma$ is ``cyclic'' in the sense that $g(s_{1/2}) = -g(s_{n+1/2})$.
  This kind of continuity and cyclicity implies property (iii) because for any integer $k \geq 1$ with $k \not\equiv_2 n$, the sequence $\gamma$ must jump over $+k$ and $-k$ an even number of times.
  Note that the same is not true for $k=0$ if $n$ is odd; in this case, the number of times that $\gamma$ jumps over $0$ is clearly odd.
  Property (iv) follows because $\gamma$ jumps over $0$ or $\pm 1$ (depending on the parity of $n$) at least once; hence, in order to jump over $\pm k$ the sequence must jump over $\pm (k-2)$ first. 
  
  In order to prove the ``if'', let us fix any vector $F = (F_0,F_1,\dots,F_{n-1})$ satisfying (i)--(iv), and let us show how to realize it by a conowheel set $P = H \cup \{w\}$ of size $n+1$.
\newcommand{\kmax}{k_{\mathsf{max}}}
Let $\kmax \geq 0$ be the largest integer that satisfies $F_{\kmax} \neq 0$.
Inspired by the first part of the proof, we construct a sequence $\gamma'$ of length $n+1$ that starts with $-\kmax-1$, ends with $+\kmax+1$, in which any two subsequent elements differ by exactly $2$, and such that the number of times $\gamma'$ jumps over any non-negative integer $k \not\equiv_2 n$ or its additive inverse $-k$ is equal to $F_k$.
Such a sequence $\gamma'$ is easy to construct, given that (i)--(iv) hold.
Indeed, one way to do it, say for even $n$, is to start with $-\kmax-1,-\kmax+1,\dots,0$ and then, for each integer $k=1,3,\dots,\kmax$ jump back and forth over $k$ exactly $F_k-1$ times.
  
  As for the construction of $P$, we start by drawing a circle with $w$ at its center, and $n$ distinct lines $s_1, \dots, s_n$ passing through $w$ in counterclockwise order.
We consider the lines to be directed in such a way that $s_2,\dots,s_n$ point into the half-space to the left of $s_1$, much like in Figure~\ref{fig:rotation}.
For each line $s_i$ we place one additional point $h$ on one of the two intersections of $s_i$ and the fixed circle around $w$.
If the $i$-th jump in $\gamma'$ is increasing, then we place the point $h$ at the back of $s_i$; otherwise, we place $h$ at the front of $s_i$.

It is clear that the resulting set $P = H \cup \{w\}$ is a conowheel set of size $n+1$.
Moreover, note that $\gamma'$ can be recovered when given only $P$ and the sequence $s_1,\dots,s_n$, simply by constructing the sequence $\gamma$ as in the first part of the proof.
Hence, $P$ has frequency vector $F$.
\end{proof}

With this characterization, it is not hard to determine the number of frequency vectors.

\thmfreqvec
\begin{proof}
  For $n=1$ and $n=2$ the formula evaluates to $1$, which is consistent with the fact that there is only one respective order type for either two or three points.
For larger $n$, we give a proof by induction for odd $n$, and note that even $n$ can be handled analogously.
  
  So, let $n = m+2 \geq 3$ be odd.
  We partition the set of frequency vectors $F = (F_0,F_1,\dots,F_{n-1})$ that are realizable by $n+1$ points into two groups, based on whether $F_0 = 1$ or not.
  \begin{itemize}
    \item If $F_0 = 1$, then $F_1 = 0$ by (ii), $F_2 \geq 2$ by (i)--(iv), and $F' = (F_2-1,F_3,F_4,\dots,F_{n-1})$ is any frequency vector realizable by $m+1$ points.
      Indeed, property (i) for $F'$ follows because the sum over all entries has decreased by two in $F'$, properties (ii) and (iii) for $F'$ follow immediately from the corresponding properties for $F$, and the same can be said for property (iv) except for the special case $F_0' = F_2-1 \neq 0$, which follows from $F_2 \geq 2$.
      
    \item If $F_0 \neq 1$, then $F_0 \geq 3$ because it is odd by (i) and (iii), $F_{n-2} = 0$ by (ii), $F_{n-1} = 0$ by (i) and (iii)--(iv), and $F'= (F_0-2,F_1,F_2,\dots,F_{n-3})$ is any frequency vector realizable by $m+1$ points.
    Indeed, properties (i)--(iii) for $F'$ follow in the same way as in the first case, and for property (iv) the only interesting special case is again $F_0' = F_0 - 2 \neq 0$, which follows from $F_0 \geq 3$.
  \end{itemize}
  If $2^{\ceil{m/2}-1}$ is the number of frequency vectors realizable by $m+1$ points, the corresponding number for $n+1$ points is thus $2 \cdot 2^{\ceil{m/2}-1} = 2^{\ceil{n/2}-1}$.
\end{proof}


\section{Geometric Graphs}
\label{sec:ggraph}

Recall that $\nb{\mathcal{G}}{P}$ is the number of crossing-free geometric graphs on a point set $P$ that are isomorphic to a graph in the family $\mathcal{G}$.
Consider now two conowheel sets $P = H \cup \{w\}$ and $P' = H \cup \{w'\}$ which can be transformed into each other by moving the extra point over the segment between $h_1,h_2 \in H$, and without encountering any other collinearities.
The situation is illustrated in Figure~\ref{fig:mutation}.
Assume that on the $w$-side of the segment $h_1h_2$ there are $i$ points of $H$, and on the $w'$-side there are $j$ points of $H$; thus, $i+j=n-2$.
Let $\delta_{i,j}$ be the increment of $\varnb{\mathcal{G}}$ when going from $P$ to $P'$, i.e., $\delta_{i,j} := \nb{\mathcal{G}}{P'} - \nb{\mathcal{G}}{P}$.

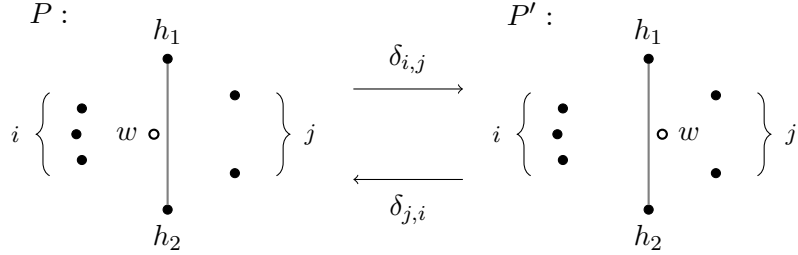
\begin{figure}
  \centering
     \begin{tikzpicture}[xscale=1.2]
       \begin{scope}[xshift=-75]
        \node[Hpt,label=above:$h_1$] (p) at (90:1) {};
         \node[Hpt,label=below:$h_2$] (q) at (270:1) {};
         \node[Wpt,label=left:$w$] (w) at (-0.15,0) {};
         \node[Hpt] at (160:1) (l1) {};
         \node[Hpt] at (180:1) (l2) {};
         \node[Hpt] at (200:1) (l3) {};
         \node[Hpt] at (35:0.9) (r1) {};
         \node[Hpt] at (-35:0.9) (r2) {};
         \draw[sptline3] (p) -- (q);
         
         \node at (-1.3,1.6) {$P:$};
         
         \draw [decorate,decoration={brace,mirror,amplitude=5pt,raise=10pt}]
(-1,0.55) -- (-1,-0.55) node [midway,xshift=-0.8cm] {\footnotesize $i$};
         \draw [decorate,decoration={brace,amplitude=5pt,raise=10pt}]
(+0.9,0.55) -- (+0.9,-0.55) node [midway,xshift=0.8cm] {\footnotesize $j$};
       \end{scope}
       
         \begin{scope}
           \draw[->] (-0.6,0.6) -- (+0.6,0.6) {};
           \draw[->] (+0.6,-0.6) -- (-0.6,-0.6) {};
           \node at (0,+1.0) {$\footnotesize\delta_{i,j}$};
           \node at (0,-1.0) {$\footnotesize\delta_{j,i}$};
         \end{scope}
       
       \begin{scope}[xshift=75]
         \node[Hpt,label=above:$h_1$] (p) at (90:1) {};
         \node[Hpt,label=below:$h_2$] (q) at (270:1) {};
         \node[Wpt,label=right:$w$] (w) at (+0.15,0) {};
         \node[Hpt] at (160:1) (l1) {};
         \node[Hpt] at (180:1) (l2) {};
         \node[Hpt] at (200:1) (l3) {};
         \node[Hpt] at (35:0.9) (r1) {};
         \node[Hpt] at (-35:0.9) (r2) {};
         \draw[sptline3] (p) -- (q);
         
         \node at (-1.3,1.6) {$P':$};
         
         \draw [decorate,decoration={brace,mirror,amplitude=5pt,raise=10pt}]
(-1,0.55) -- (-1,-0.55) node [midway,xshift=-0.8cm] {\footnotesize $i$};
         \draw [decorate,decoration={brace,amplitude=5pt,raise=10pt}]
(+0.9,0.55) -- (+0.9,-0.55) node [midway,xshift=0.8cm] {\footnotesize $j$};
       \end{scope}
     \end{tikzpicture}
  \caption{Moving the extra point over the segment $h_1h_2$ for the case $i=3$ and $j=2$.}
  \label{fig:mutation}
\end{figure}

\begin{lemma}
  \label{lem:mutation}
  For every ${\cal G}$, $\delta_{i,j}$ is well-defined; i.e., its value depends only on $i$, $j$ and $\mathcal{G}$, and not on the exact placement of $H$ or the location where the extra point traverses the segment between $h_1$ and $h_2$.
\end{lemma}
\begin{proof}
  All geometric graphs that do not contain the edge $\{h_1,h_2\}$ are not affected by the mutation; i.e., they are crossing-free on $P$ if and only if they are crossing-free on $P'$.
  Therefore, $\delta_{i,j}$ is equal to the number of crossing-free geometric graphs on $P'$ containing $\{h_1,h_2\}$ minus the number of crossing-free geometric graphs on $P$ containing $\{h_1,h_2\}$.
  For the following reasons, these numbers depend only on $i$, $j$ and $\mathcal{G}$.
  
  In the case of $P$, on the $w$-side we have $i+3$ points (including $h_1$ and $h_2$) in a barely-in configuration, for which there exists a unique order type.
  On the opposite side we have $j+2$ points (including $h_1$ and $h_2$) in convex position, for which there also exists a unique order type.
  Because of the presence of the edge $\{h_1,h_2\}$ between two extreme points, any other edges must be completely contained in one of the two sides, and the claim follows.
  In the case of $P'$, an analogous argument works.
\end{proof}

\subparagraph*{Example, embracing triangles.} Consider the case ${\mathcal{G}} = \{\fillup{K_4}\}$.
Observe that any crossing-free embedding of $\fillup{K_4}$ on $P$ uses $w$ as the inner vertex of the underlying $K_4$.
Furthermore, if the embedding uses the edge $\{h_1,h_2\}$, which implies that $h_1$ and $h_2$ are outer vertices of $K_4$, then any one of the $i$ points on the left hand side can be used as the third outer vertex of $K_4$.
This gives exactly $i$ crossing-free embeddings of $\fillup{K_4}$ on $P$ which use the edge $\{h_1,h_2\}$.
Similarly, we get $j$ for the corresponding number of embeddings on $P'$.
Therefore, $\delta_{i,j} = j-i$ for $\mathcal{G} = \{\fillup{K_4}\}$.
\hfill \break
\thmggraph
\begin{proof}
  We proceed by choosing the numbers $\lambda_{l,r}$ such that the validity of the formula is preserved under continuous motion of $P$, and then choose $\gamma_n$ such that it holds for some starting configuration.
  To be more concrete, we allow continuous motion of $P$ as long as it remains a conowheel set.
  At discrete moments in time we allow collinearity of three points (but not more), the one in the middle being $w$.
  As already seen in Figure~\ref{fig:heptagon} for the special case $n=7$, any two conowheel sets can be transformed into each other in this manner.
  
  Let now $P$ and $P'$ be as in Figure~\ref{fig:mutation}.
  Note that the values $\lside{h}$ and $\rside{h}$ do not change for any $h \in H \setminus \{h_1,h_2\}$ when going from $P$ to $P'$, since there are no other collinearities than the one involving $h_1$ and $h_2$.
  For $h_1$ and $h_2$ the corresponding values are 
  \begin{align*}
    & P:  && \lside{h_1} = \rside{h_2} = i   && \rside{h_1} = \lside{h_2} = j+1 \enspace, \\
    & P': && \lside{h_1} = \rside{h_2} = i+1 && \rside{h_1} = \lside{h_2} = j \enspace.
  \end{align*}
  Evaluating the formula in Theorem~\ref{thm:ggraph} for $P'$ and $P$, and then taking the difference yields
  \begin{equation*}
    \delta_{i,j} \overset{!}{=} \frac{1}{2} (\lambda_{i,j+1} + \lambda_{j+1,i}) - \frac{1}{2} (\lambda_{i+1,j} + \lambda_{j,i+1}) =
    \lambda_{i,j+1} - \lambda_{i+1,j} \enspace.
  \end{equation*}
  
  That is, we preserve the validity of the formula under continuous motion as long as the numbers $\lambda_{l,r}$ are chosen in such a way that the above equality holds.
  Setting $\lambda_{l,r} := \delta_{n-2,0} + \delta_{n-3,1} + \dots + \delta_{l,r-1} + c_n$ satisfies this constraint.
  Moreover, the symmetry $\lambda_{l,r} = \lambda_{r,l}$ (which is needed for simplifying the right hand side in the above equation) indeed follows from the skew-symmetry $\delta_{i,j} = - \delta_{j,i}$, which guarantees that the additional summands (say, in $\lambda_{l,r}$ when compared to $\lambda_{r,l}$) cancel out.
 Further note that $c_n$ is an arbitrary integer independent of $l$ and~$r$.
 For the proof to go through one could simply fix $c_n=0$, but in some of the applications other values for $c_n$ will turn out to be more convenient.
  
  Finally, $\gamma_n$ is chosen in such a way that the formula holds for some conowheel set.
  The most natural choice for ``anchoring'' the formula is a set in convex position.
  Putting
  \begin{equation*}
    \gamma_n :=
      \nb{\mathcal{G}}{\Pconv} + \frac{1}{2} \sum_{{l,r \colon l + r = n - 1}} \lambda_{l,r}
  \end{equation*}
  works, since for a set in convex position the sum in the formula of Theorem~\ref{thm:ggraph} and the sum in the definition of $\gamma_n$ cancel out.
\end{proof}

As already noted in the introduction, computing the frequency vector can be done in $O(n \log n)$ time.
Given the values $\Lambda_k$, computing the number $\nb{\mathcal{G}}{P}$ of embeddings then requires only $O(n)$ additional arithmetic operations.

\subparagraph*{Example continued, embracing triangles.}
We have already derived $\delta_{i,j} = j-i$ for $\mathcal{G} = \{\fillup{K_4}\}$.
This now gives rise to
\begin{equation*}
  \lambda_{l,r} = \delta_{n-2,0} + \delta_{n-3,1} + \dots + \delta_{l,r-1} + c_n = \sum_{j=0}^{r-1} j - \left(\binom{n-1}{2} - \sum_{i=0}^{l-1} i\right) + c_n = \binom{l}{2} + \binom{r}{2} \enspace ,
\end{equation*}
if we choose $c_n = \binom{n-1}{2}$.
It can be checked that $\nb{\mathcal{G}}{\Pconv} = 0$ or, in words, that there is no way to draw a $w$-embracing triangle spanned by $H$ if $w$ is not contained in the convex hull of~$H$.
Hence,
\begin{equation*}
  \gamma_n = \nb{\mathcal{G}}{\Pconv} + \frac{1}{2} \sum_{{l,r \colon l + r = n - 1}} \lambda_{l,r} = 0 + \frac{1}{2}\sum_{l=0}^{n-1} \binom{l}{2} + \frac{1}{2}\sum_{r=0}^{n-1} \binom{r}{2} = \binom{n}{3} \enspace .
\end{equation*}
After putting everything together we obtain the exact formula displayed earlier in equation~(\ref{eq:embtrg}).

\subsection{Further Examples}
We call the following two simple applications ``insensitive'' since the number of crossing-free embeddings is the same on all wheel sets, but may be different for sets in convex position.

\subparagraph*{Spanning cycles.}
Consider the case where $\mathcal{G}$ consists of a cycle over $n+1$ vertices.
Note that if $P$ is in convex position, then there is only one spanning cycle, whereas if $P$ is a barely-in wheel set, then there are $n$ spanning cycles. 
Hence, we get $\delta_{0,n-2} = -\delta_{n-2,0} = n-1$.
Moreover, for all other mutations we get $\delta_{i,j} = 0$, since $i$ and $j$ being non-zero implies that no crossing-free spanning cycle can ever contain the edge $\{h_1,h_2\}$ as in Figure~\ref{fig:mutation}.
It follows that all wheel sets over $n+1$ points admit $n$ crossing-free spanning cycles (which can easily be seen directly).

\subparagraph*{Spanning paths.}
If $\mathcal{G}$ consists of a path over $n+1$ vertices we also get $\delta_{i,j} = 0$ unless $i=0$ or $j=0$, but for a different reason.
Consider again Figure~\ref{fig:mutation}.
On $P$ there are $2 \cdot 2^i \cdot 2^{j-1}$ crossing-free embeddings that use the edge $\{h_1,h_2\}$; the factor $2$ is for deciding which one of $h_1$ and $h_2$ is connected to the left hand side, the factor $2^i$ is for completing the left hand side (a subset of $i+2$ points in convex position including $h_1$, say) to a path, and the factor $2^{j-1}$ is for completing the right hand side (a subset of $j+1$ points in convex position including $h_2$, say) to a path.
Likewise, on $P'$ there are $2 \cdot 2^{i-1} \cdot 2^j$ embeddings, which is the same number.

For the special case $i=0$ and $j=n-2$, we get $2^{n-2}$ embeddings on $P$ that use the edge $\{h_1,h_2\}$ using a similar way of counting as above.
For $P'$ we require a slightly more involved argument, however.
All remaining points except $h_1$ and $h_2$ are first divided in two parts based on whether they belong to the path extension that starts in $h_1$ or $h_2$.
There are only $2 \cdot (n-1)$ sensible such divisions; the factor $2$ is for deciding whether $w$ belongs to the path starting in $h_1$ or $h_2$, and the factor $(n-1)$ is for deciding how many of the remaining $n-2$ points belong to $h_1$ or $h_2$ (this is enough because in order to obtain a crossing-free spanning path, both subsets of points must appear consecutively along the boundary of the convex hull of $P'$).
Now suppose for example that the extra point $w$ is contained in the path extension that starts in $h_2$, and that $k \geq 1$ other points (from the remaining $n-2$) belong to $h_1$.
Under these assumptions, we can again compute that there are $2^{k-1} \cdot 2^{n-2-k}$ crossing-free embeddings since both parts are in convex position and do not interfere with each other.
The total number of embeddings on $P'$ is obtained by summing over all values of $k$ (beware of the special case $k=0$), which finally yields
\begin{equation*}
  \delta_{0,n-2} = -\delta_{n-2,0} = 2 \cdot (2^{n-2} + \sum_{k=1}^{n-2} 2^{k-1} \cdot 2^{n-2-k}) - 2^{n-2} = (n-1) 2^{n-2} \enspace .
\end{equation*}

For anchoring we can compute $\nb{\mathcal{G}}{\Pconv} = (n+1)2^{n-2}$.
Adding up the two numbers that we computed last yields $\nb{\mathcal{G}}{P} = n2^{n-1}$ for all wheel sets $P$.

\hfill \break

The following two applications are ``sensitive'' in the sense that different wheel sets in general have different numbers of crossing-free embeddings.
The running example with $w$-embracing triangles also is of this kind.

\subparagraph*{Matchings.}
Let $\mathcal{G} = \mathcal{M}$, the set of (not necessarily perfect) matchings over $n+1$ vertices.
It is known that $\nb{\mathcal{M}}{\Pconv} = M_{n+1} := \sum_{k=0}^{\floor{(n+1)/2}} \binom{n+1}{2k}C_k$, the $(n+1)$-th Motzkin number\footnote{Sequence A001006 on OEIS.} \cite{Motzkin48}, which is defined in terms of the Catalan numbers\footnote{Sequence A000108 on OEIS, which also counts abstract binary trees.} $C_k := \frac{1}{k+1}\binom{2k}{k}$.
It is thus easy to compute $\delta_{i,j} = M_iM_{j+1} - M_{i+1}M_j$ since, as always, we have to worry only about embeddings that use the edge $\{h_1,h_2\}$.
This gives $\lambda_{l,r} = M_lM_r$ and $\gamma_n = M_{n+1} + \frac{1}{2} \sum_{l,r} M_lM_r$.
We can further get rid of the sum in $\gamma_n$ by using the equation $M_{n+1} = M_n + \sum_{l,r} M_lM_r$, which holds because a crossing-free matching on $\Pconv$ ($M_{n+1}$ possibilities) either leaves $w$ unmatched ($M_n$ possibilities) or it matches $w$ with one of the other $n$ points ($\sum_{l,r} M_lM_r$ possibilities).
Hence,
\begin{equation}
  \nb{\mathcal{G}}{P} = \frac{3M_{n+1} - M_n}{2} - \frac{1}{2} \sum_{h \in H} M_{\lside{h}}M_{\rside{h}} \hspace{1cm} \text{for $\mathcal{G} = \mathcal{M}$} \enspace.
\end{equation}

\subparagraph*{Spanning trees.}
Let $\mathcal{G} = \mathcal{ST}$, the set of all trees over $n+1$ vertices.
We will make use of the fact that $\nb{\mathcal{ST}}{\Pconv} = T_{n+1} := \frac{1}{2n+1} \binom{3n}{n}$ \cite{DulucqP93,FlajoletN99}.\footnote{Sequence A001764 on OEIS, which, incidentally, also counts abstract ternary trees.}
Furthermore, we will use the hypergeometric identity $\sum_{k=0}^n T_{k+1}T_{n-k+1} = \frac{1}{n+1} \binom{3n+1}{n}$.\footnote{Sequence A006013 on OEIS, which also counts \emph{pairs} of abstract ternary trees, implying the used identity since the left hand side is just the convolution of the counting sequence of abstract ternary trees.}

To compute $\delta_{i,j}$, consider the set $P$ as in Figure \ref{fig:mutation}.
In order to complete the left hand side into a spanning tree, we have to build two disjoint trees rooted at $h_1$ and $h_2$, respectively.
There are $2$ choices for assigning $w$ either to the upper or the lower tree, and there are $i+1$ choices for distributing the $i$ points on the left among the two trees.
Indeed, the latter claim holds because the $k$ out of $i$ points assigned to $h_1$, say, have to appear consecutively with $h_1$ on the convex hull as otherwise we are forced to create a crossing.
Once the points have been distributed, we are left with two point sets of size $k+1$ and $i-k+2$ in convex position.
For completing the right hand side into a spanning tree, a simpler argument can be used without the additional complication of $w$.
Moreover, the set $P'$ can be handled analogously, which yields
\begin{equation}\label{eq:delta_spanning_tree}
\begin{aligned}
      \delta_{i,j}
  & = 2 \sum_{k=0}^j T_{k+1}T_{j-k+2} \cdot \sum_{k=0}^i T_{k+1}T_{i-k+1}
    - 2 \sum_{k=0}^i T_{k+1}T_{i-k+2} \cdot \sum_{k=0}^j T_{k+1}T_{j-k+1} \\
  & = 2 \left( \frac{2}{j+2}\binom{3j+3}{j} \cdot \frac{1}{i+1}\binom{3i+1}{i}
      - \frac{2}{i+2}\binom{3i+3}{i} \cdot \frac{1}{j+1}\binom{3j+1}{j}\right) \enspace.
\end{aligned}
\end{equation}

For this application, we do not know if a simple closed form expression for $\lambda_{l,r}$ exists.
Still, note that if one were to compute $\nb{\mathcal{ST}}{P}$, the numbers $\delta_{i,j}$ could be summed up using $O(n)$ arithmetic operations and the value of $\gamma_n$ could be computed on the fly for any given $n$.

\subparagraph*{Related applications.}
Observe that, for example, a geometric triangulation of $\Pconv$ can be embedded as a plane graph on $\Pbarely$.
However, this embedding is no longer a triangulation (i.e., a tessellation of the convex hull of $\Pbarely$ into triangles).
Hence, there is no natural choice of $\mathcal{G}$ such that $\nb{\mathcal{G}}{P}$ is the number of triangulations of any conowheel set~$P$.

Nevertheless, the continuous motion argument is still applicable and leads to a similarly simple formula.
All that is required is an adapted version of Lemma~\ref{lem:mutation} that shows that the values $\delta_{i,j}$ (i.e., the increase in the number of triangulations as $w$ moves about) are still well-defined and only dependent on $i$ and $j$.
From the description in \cite{Randall01} it follows that $\delta_{i,j} = C_{i}C_{j+1} - C_{i+1}C_{j}$, from which we get
\begin{equation}
  \mathop{\mathsf{tr}}(P) = \frac{C_n}{2} - \frac{1}{2} \sum_{h \in H} C_{\lside{h}}C_{\rside{h}} 
\end{equation}
for the number of triangulations of a conowheel set $P$ by using the techniques developed earlier in this section.
The above formula has been mentioned already in \cite{RuizW15}.

There are several other families of geometric graphs (pseudo-triangulations, crossing-free convex partitions, etc.)
whose quantity on a conowheel set $P$ cannot be expressed easily in the form $\nb{\mathcal{G}}{P}$, 
but for which it is also possible to adapt Lemma~\ref{lem:mutation}.
We provide the example for crossing-free convex partitions in Appendix~\ref{sec:convex_partitions}.

Furthermore, we note that Lemma~\ref{lem:mutation} and, hence, Theorem~\ref{thm:ggraph} generalizes to crossing-free embeddings of hypergraphs, where ``crossing-free'' means that the convex hulls of any two hyperedges intersect in an at most 1-dimensional set.
In similar fashion to the original proof, we observe that an embedding of a hypergraph is crossing-free on $P$ if and only if it is crossing-free on $P'$, as long as there is no hyperedge that contains both $h_1$ and $h_2$.
Therefore, we again get a clean separation between the left and right hand side for all embeddings that matter.

\subsection{The Symmetric Configuration Maximizes}

For many families of crossing-free geometric graphs, it is known that a set of points in convex position minimizes their number.
In particular, this is true for the family of all crossing-free geometric graphs, connected crossing-free geometric graphs, as well as for any family of cycle-free graphs such as (perfect) matchings or spanning trees~\cite{AichholzerHHHKV07}.
Triangulations, on the other hand, are a well-known counter-example for which it is known that convex position does not minimize~\cite{almost_convex}.

Not much seems to be known about point sets that maximize the number of crossing-free geometric graphs;
there are merely constructions that give lower bounds on this number~\cite{AichholzerHHHKV07,bounds_multiplicity,production}.
Nevertheless, in what follows we give a sufficient condition which allows us to prove that among all conowheel sets, the symmetric wheel set $\Psym$ maximizes the number of many families of crossing-free geometric graphs.
Recall that $\Psym$ is constructed by taking the vertex set of a regular $n$-gon and by adding an extra point $w$ at the center of that $n$-gon.
For the case that $n$ is even, we slightly perturb $w$ in order to obtain a point set in general position.
Irrespective of the perturbation we get the same order type and, thus, the same frequency vector.

\begin{lemma}\label{lem:max}
  Let $\mathcal{G}$ be any family of graphs for which $\delta_{i,j} \geq 0$ holds for all $i < j$.
  Then, the number $\nb{\mathcal{G}}{P}$ is maximized for $P = \Psym$.
\end{lemma}
\begin{proof}
  We start with a point set that contains the vertices of a regular $n$-gon and an extra point $w$ added in such a way that the whole set is in convex position.
  Naturally, the number of non-crossing embeddings of $\mathcal{G}$ on this set is $\nb{\mathcal{G}}{\Pconv}$.
  
  The idea now is to let $w$ move on a straight line towards its final position in $\Psym$ and to observe how the number of non-crossing embeddings changes.
  We obtain
  \begin{align*}
    \nb{\mathcal{G}}{\Psym} = \nb{\mathcal{G}}{\Pconv} + \sum_{\substack{i,j \colon i < j \\ i+j = n-2}} (i+1) \cdot\delta_{i,j} \enspace,
  \end{align*}
  where the factor $i+1$ is due to the fact that we can draw $i+1$ segments with endpoints in $H$ that contain $w$ and $i$ other points from $H$ on one side, thus giving rise to exactly $i+1$ transitions that increase the number of embeddings by $\delta_{i,j}$.
  Also note that, for the case that $n$ is even, we always have $\delta_{i,j} = 0$ whenever $i=j$ by symmetry.
  Hence, we need not worry about how exactly $w$ was perturbed in $\Psym$.
  
  Let now $P$ be an arbitrary conowheel set, for which we wish to prove that $\nb{\mathcal{G}}{P} \leq \nb{\mathcal{G}}{\Psym}$ holds.
  Here, we start with a point set that is a copy of $P$ except that the extra point $w$ has been moved in such a way that the whole set is again in convex position.
  In the same way as before, we let $w$ move on a straight line to its final position in $P$.
  Here, we obtain
  \begin{align*}
    \nb{\mathcal{G}}{P} = \nb{\mathcal{G}}{\Pconv} +
    \sum_{\substack{i,j \colon i < j \\ i+j = n-2}} \enspace \underbrace{\alpha(i)}_{\leq i+1} \cdot \underbrace{\delta_{i,j}}_{\geq 0} +
    \sum_{\substack{i,j \colon i > j \\ i+j = n-2}} \alpha(i) \cdot \underbrace{\delta_{i,j}}_{\leq 0}
    \leq \nb{\mathcal{G}}{\Psym} \enspace.
  \end{align*}
  In the formula above, the number $\alpha(i)$ indicates, for any fixed $i$, how often a $(i,j)$-transition occurs during this process.
  The inequality $\alpha(i) \leq i+1$ for $i<j$ follows again because we can draw $i+1$ segments with endpoints in $H$ that contain $w$ and $i$ other points from $H$ on one side, thus giving rise to at most $i+1$ transitions that increase the number of embeddings by $\delta_{i,j}$.
\end{proof}

  It is easy to see that the condition of Lemma~\ref{lem:max} holds for $w$-embracing triangles.
  By making use of the fact that Motzkin numbers are log-concave, i.e., $M_i^2 \geq M_{i-1}M_{i+1}$ \cite{Aigner98}, it can also be shown to hold for crossing-free matchings.
  For crossing-free spanning trees, on the other hand, starting from equation~(\ref{eq:delta_spanning_tree}) for $\delta_{i,j}$ one can derive the equivalence
  \begin{align*}
    \delta_{i,j} \geq 0 \; \Longleftrightarrow \; \frac{(i+1)(i+\frac{2}{3})}{(i+2)(i+\frac{3}{2})} \leq \frac{(j+1)(j+\frac{2}{3})}{(j+2)(j+\frac{3}{2})} \enspace.
  \end{align*}
  Assuming $i < j$, we see that the factor $\frac{i+1}{i+2}$ is dominated by the factor $\frac{j+1}{j+2}$.
  The same can be said for the other two factors, and hence we conclude $\delta_{i,j} \geq 0$.

  \begin{corollary}
    When restricted to conowheel sets, the numbers of $w$-embracing triangles, crossing-free matchings and crossing-free spanning trees are maximized for $\Psym$.
  \end{corollary}
  
  Note that the above does not hold for crossing-free perfect matchings.
  This special case has been analyzed in \cite{RuizW15}.


\subsection{Wheel Sets and the Rectilinear Crossing Number}\label{sec:crossing_number}

Even though conowheel sets and the associated frequency vectors seem like a very specific set of objects, we show here that they sometimes can be useful also in a more general setting, where we are given an arbitrary point set in the plane.

Consider for example any set $\tilde P$ of $n+1$ points in general position and let $\Box$ and $\triangle$ be the number of $4$-element subsets of $\tilde P$ in convex and in non-convex position, respectively.%
\footnote{The number $\Box$ is also the number of crossings of the complete geometric graph on $\tilde P$, a quantity that has obtained special attention in connection with the so-called \emph{rectilinear crossing number of $K_n$} (i.e., the smallest number of crossings in a straight line drawing of the complete graph in the plane).}
Let $p \in \tilde P$ be any point.
We can construct a conowheel set $P = H \cup \{w\}$ that contains $w = p$ and, for every $q \in \tilde P \setminus \{p\}$, the point $h$ which lies on the intersection of the ray from $p$ to $q$ and a fixed circle centered at $w$ (as done also, e.g., in~\cite{khuller_mitchell}).
That is, $P$ is simply a representation of the local sequence \cite{semispaces} of $p$ in $\tilde P$ in terms of conowheel sets.
Further observe that a triangle spanned by points in $\tilde P$ contains $p$ if and only if its image in $P$ contains $w$.
Hence, $\nb{\{\fillup{K_4}\}}{P}$ is the number of such triangles, which is given by equation~(\ref{eq:embtrg}).
We thus obtain $\triangle$ by a summation over all points $p$ in~$\tilde P$.
Since $\Box + \triangle = \binom{n+1}{4}$, we also obtain~$\Box$.

To sum up, we can associate a frequency vector with every point of a given point set, and this set of frequency vectors then determines the value of~$\Box$.
Unfortunately, this simple approach does not work in general for other interesting quantities;
there are examples of point sets with the same set of frequency vectors but a different number of triangulations, see Appendix~\ref{sec:counterexamples1}.

As a further side remark, we note that what we did above gives rise to the same equations from~\cite{AbregoF05,LovaszVWW04,Wagner03} that express $\Box$ in terms of the number of $j$-edges (i.e., directed edges spanned by $\tilde P$ with exactly $j$ points of $\tilde P$ to their left).


\section{Embracing Sets}
\label{sec:originembracing}

It turns out that $w$-embracing triangles, the running example from the previous section, offer more than meets the eye at first sight.
We show here that mere information about the structure (or number) of $w$-embracing triangles (or larger $w$-embracing sets) is enough to uniquely determine the order type (or frequency vector) of a conowheel set.
Ultimately, this leads to a converse of Theorem~\ref{thm:ggraph}.


\subsection{Embracing Triangles determine Order Type}
\label{sec:embrordtype}

By the \emph{family of $w$-embracing triangles} in a set $H = \{h_1,\dots,h_n\}$ we mean the set $T \subseteq \binom{[n]}{3}$ where $\{i,j,k\}$ is contained in $T$ if and only if $h_i,h_j,h_k$ span a $w$-embracing triangle.
For ease of notation, we identify the index set $[n] = \{1,\dots,n\}$ with $H$ and write $ijk$ short for an unordered triple in $T$.

\begin{lemma}\label{lem:trgs_det_ordtype}
Let $P = H \cup \{w\}$ be a conowheel set.
Then, the family $T$ of $w$-embracing triangles in $H$ determines the order type of $P$.
\end{lemma}
\begin{proof}
If there are no $w$-embracing triangles in $T$, then $P$ is in convex position and we are done.
Thus, we may assume that $w$ is not an extreme point.
We may now fix any $w$-embracing triangle $abc$ in $T$.
Without loss of generality, let the triangle $abc$ be oriented counterclockwise; otherwise, relabel the points.
We will determine the orientation of all other point triples in $P$ with respect to the orientation of $abc$.

Let $pqr$ be any other $w$-embracing triangle in $T$ for which we wish to determine the orientation.
Without loss of generality, $pqr$ does not share any vertices with $abc$; if it does, simply duplicate and slightly perturb the corresponding vertices in $pqr$ while maintaining the wheel set property.
Observe now, in Figure~\ref{fig:trgs_det_ordtype_cases}, that in any case there must be a third $w$-embracing triangle, say $apq$, that has one vertex of the already oriented triangle and two vertices of the other.
Observe further that the triangle $apq$ is oriented counterclockwise if and only if at least one of $abq$ and $apc$ is a $w$-embracing triangle.
Having obtained the orientation of the edge $pq$ with respect to $w$, it is now easy to orient the original $w$-embracing triangle $pqr$.
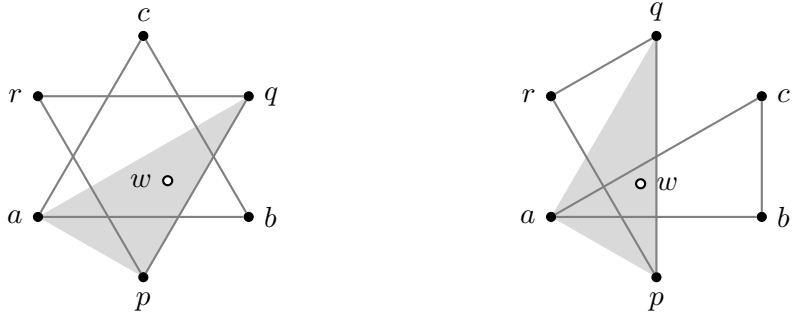
\begin{figure}
  \begin{center}
    \begin{tikzpicture}[scale=1.6]
      \begin{scope}[xshift=0]
        \node[Hpt,label=left:$a$] (a) at (210:1) {};
        \node[Hpt,label=right:$b$] (b) at (330:1) {};
        \node[Hpt,label=above:$c$] (c) at (90:1) {};

        \node[Hpt,label=below:$p$] (p) at (270:1) {};
        \node[Hpt,label=right:$q$] (q) at (30:1) {};
        \node[Hpt,label=left:$r$] (r) at (150:1) {};
        
        \fill[black,opacity=0.15] (a.center) -- (p.center) -- (q.center) -- cycle;
        
        \node[Wpt,label=left:$w$] (w) at (0.2,-0.2) {};
        
        \draw[sptline3] (a) -- (b);
        \draw[sptline3] (b) -- (c);
        \draw[sptline3] (c) -- (a);
        
        \draw[sptline3] (p) -- (q);
        \draw[sptline3] (q) -- (r);
        \draw[sptline3] (r) -- (p);
      \end{scope}
      \begin{scope}[xshift=120]
        \node[Hpt,label=left:$a$] (a) at (210:1) {};
        \node[Hpt,label=right:$b$] (b) at (330:1) {};
        \node[Hpt,label=right:$c$] (c) at (30:1) {};

        \node[Hpt,label=below:$p$] (p) at (270:1) {};
        \node[Hpt,label=above:$q$] (q) at (90:1) {};
        \node[Hpt,label=left:$r$] (r) at (150:1) {};
        
        \fill[black,opacity=0.15] (a.center) -- (p.center) -- (q.center) -- cycle;
        
        \node[Wpt,label=right:$w$] (w) at (240:0.26) {};
        
        \draw[sptline3] (a) -- (b);
        \draw[sptline3] (b) -- (c);
        \draw[sptline3] (c) -- (a);
        
        \draw[sptline3] (p) -- (q);
        \draw[sptline3] (q) -- (r);
        \draw[sptline3] (r) -- (p);
      \end{scope}
    \end{tikzpicture}
  \end{center}
  \caption{Two $w$-embracing triangles $abc$ and $pqr$, as in the proof of Lemma~\ref{lem:trgs_det_ordtype}. In the first case, either $apq$ or $aqr$ is also $w$-embracing; in the second case, $apq$ is also $w$-embracing.}
  \label{fig:trgs_det_ordtype_cases}
\end{figure}

Before continuing with the remaining unoriented triangles, we define the following equivalence relation $\sim$ over the index set $[n]$.
We put $i \sim i'$ if and only if
\begin{equation*}
  \forall j,k \in [n] \colon \enspace ijk \in T \Longleftrightarrow  i'jk \in T \enspace .
\end{equation*}
Intuitively, equivalence classes defined by $\sim$ are sequences of points that appear consecutively in the radial order around $w$ without any points on the other side; or, referring back to Figure~\ref{fig:bijection}, consecutive sequences of black beads in the corresponding necklace.

The already determined orientation of the $w$-embracing triangles in $T$ yields the relative position of $w$ and any two such equivalence classes; indeed, given representatives $i$ and $j$ of two distinct equivalence classes, one can always find a third point $k$ to form a $w$-embracing triangle.
Hence, we have determined the order of the defined equivalence classes along the boundary of the convex hull of $P$.

It only remains to determine the internal order of the points in any given equivalence class along the boundary of the convex hull.
However, since between any two possible choices of such internal orders there is an orientation-preserving bijection, all resulting point sets have the same order type.
\end{proof}


\subsection{Numbers of Embracing Sets determine Frequency Vector}
\label{sec:embrfreqvec}

Let us generalize the notion of $w$-embracing triangles to larger sets.
A subset $A \subseteq H$ is a \emph{$w$-embracing $k$-set} if $w$ is in the convex hull of $A$ and $|A| = k$.
For simplicity, we consider here only $w = \mathbf{0}$ and call $A$ an \emph{origin-embracing $k$-set}, or simply \emph{embracing $k$-set}.

We can show that the number of embracing $k$-sets is determined by the frequency vector of a conowheel set $P = H \cup \{w\}$ for any $k$, and not just for $k=3$ as seen earlier in equation~(\ref{eq:embtrg}).
Indeed, since $H$ is in general position, for every non-embracing $k$-set $A \subseteq H$ there exists a unique point $h \in A$ such that the convex hull of $A$ is in the closed halfplane to the left of the directed line through $w$ and $h$.
For any possible choice of that point $h \in H$ we can construct $\binom{\lside{h}}{k-1}$ such non-embracing $k$-sets, and thus it is possible to get a generalization of equation (\ref{eq:embtrg}) for $\embr{k}$, the number of embracing $k$-sets:
\begin{equation}\label{eq:embracing_plane}
\embr{k} = \binom{n}{k} - \sum_{h \in H} \binom{\lside{h}}{k-1} \enspace.
\end{equation}

Interestingly, the converse also turns out to be true.

\begin{lemma}
\label{lem:embracing_implies_frequency}
Let $P = H \cup \{w\}$ be a conowheel set of size $n+1$.
Then, the sequence $(\embr{k})_{k = 3}^n$ determines the frequency vector of $P$.
\end{lemma}
\begin{proof}
Let $E = (\embr{k})_{k=3}^n$. 
Consider the vector $L = (L_j)_{j=1}^{n-1}$ where $L_j$ is the number of points $h \in H$ with $\lside{h} = j$.
Clearly, $L$ determines the frequency vector $F(P)$.
It thus suffices to show that $E$ determines $L$.

Note that equation (\ref{eq:embracing_plane}) may now be rewritten as
\begin{equation*}
  \binom{n}{k} - \embr{k} = \sum_{j=1}^{n-1} L_j \binom{j}{k-1} \enspace.
\end{equation*}
Further note that the above equation also holds for $k = 2$, since in that case both sides of the equation count the number of pairs of points in $H$.
We can thus define a vector $E' = (e_i)_{i=1}^{n-1}$ with $e_i = \binom{n}{i+1} - \embr{i+1}$ and a square matrix $A = (a_{i,j})_{i,j=1}^{n-1}$ with $a_{i,j} = \binom{j}{i}$, such that the equality $E' = A L$ holds.
Since the matrix $A$ is unitriangular, it has an inverse, from which we can conclude that $E'$ determines $L$.
\end{proof}

\thmggraphconverse
\begin{proof}
We already know from Theorem~\ref{thm:ggraph} that the frequency vector determines the number of non-crossing embeddings.
For the other direction, the idea is to reconstruct the numbers of embracing $k$-sets by appropriately choosing the graph classes $\mathcal{G}$ for every $k$.
After that, the frequency vector is determined by Lemma~\ref{lem:embracing_implies_frequency}.

The number $\embr{3}$ of $w$-embracing triangles is equal to the number $\nb{\{\fillup{K_4}\}}{P}$ and therefore, by assumption, the same for both $P$ and $P'$.
We now simply generalize to $k$-wheels; that is, we consider a set $\mathcal{G}$ that contains a single cycle with $k$ vertices, each adjacent to one additional vertex.
All that is left to observe is that for such $\mathcal{G}$ the number $\embr{k}$ of embracing $k$-sets is equal to $\nb{\mathcal{G}}{P}$, and hence the same for both $P$ and $P'$.
\end{proof}

As a side remark, note that for arbitrary point sets we can compute the number of crossing-free embeddings of such $k$-wheels (as in the above proof) in polynomial time.
For $k=3$, this number is equal to the number of non-crossing embeddings of $\fillup{K_4}$, which can be obtained from the set of frequency vectors associated with each point, see Section~\ref{sec:crossing_number}.
For $k \geq 4$, we distinguish the cases where the geometric embedding of a $k$-wheel has only three vertices on the unbounded cell, and where it has $k$.
The latter case can be dealt with by computing the number of embracing $k$-sets for each point.
The former can be obtained by computing, for every integer $i$, the number of triangles with $i$ points in the interior, and then multiplying this number by $3\binom{i}{k-2}$.
This is because for every vertex of such a triangle, a path of $k-2$ points inside this triangle in radial order around that vertex gives a $k$-wheel with the triangle as the unbounded cell.
For each $i$, the corresponding number of triangles can be obtained in $O(n^3)$ time~\cite{akm93,eorw92}.


\section{Higher Dimensions: Origin-Embracing Simplices}
\label{sec:higher_dimensions}

As already noted in the introduction, the concept of conowheel sets can be generalized to higher dimensions.
However, already in $\R^3$ we face certain challenges.
For example, the number of tetrahedralizations of $n+1$ points in convex position in $\R^3$ does not only depend on~$n$, in contrast to the 2-dimensional case.
Even when considering simpler structures, like the set of $w$-embracing tetrahedra, the ideas from Section~\ref{sec:ggraph} do not generalize.
Intuitively, our argument of $w$ moving over a segment does not work in 3-space, as it can move ``around'' a triangle.

Still, we show here how to use similar ideas in order to obtain improved time bounds for computing the simplicial depth of a point~$w$ in a set $H$.

\subparagraph*{Oriented simplices.} Given a set $T$ of $d$ affinely independent points in $\R^d$, its convex hull $\conv{T}$ is a $(d-1)$-simplex and its affine hull is a hyperplane.
We want to be able to refer to the two sides of this hyperplane by identifying a positive and a negative side.
For that consider a sequence $p_1p_2\ldots p_d$ of $d$ affinely independent points.
The affine hull of $T=\{p_1,p_2\commadots p_d\}$ can be described as the set of points $q$ with $\side(q,p_1p_2\ldots p_d) = 0$, where 
\begin{equation*}
\side(q,p_1p_2\ldots p_d) :=\det(p_1-q,p_2-q\commadots p_d-q) \enspace .
\end{equation*}
We call the set of points $q$ with $\side(q,p_1p_2\ldots p_d) >0$ the \emph{positive side of $p_1p_2\ldots p_d$}, and the set of points $q$ with $\side(q,p_1p_2\ldots p_d) <0$ the \emph{negative side of $p_1p_2\ldots p_d$}.

%

An \emph{oriented $(d-1)$-simplex} is a sequence $p_1p_2\ldots p_d$ of $d$ affinely independent points, with its associated $(d-1)$-simplex, and its associated positive and negative side as above.
Two such oriented $(d-1)$-simplices are equivalent if their sequences can be obtained from each other by an even number of transpositions (e.g., $p_1 p_2 p_3$, $p_3 p_1 p_2$, and $p_2 p_ 3 p_1$ are equivalent).

Via oriented simplices, the concept of an order type generalizes to arbitrary dimensions.
Similar to the 2-dimensional case, the order type of a conowheel set $P = H \cup \{w\}$ in $\R^d$ determines the set of points on the positive side of the oriented $(d-1)$-simplex $w h_1 \dots h_{d-1}$, for each $(d-1)$-tuple in $H$.
We denote by $\lside{h_1 \dots h_{d-1}}$ the number of points on that positive side, and we denote by $\rside{h_1 \dots h_{d-1}}$ the number of points on the negative side.
We can thus define the frequency vector $F(P) = (F_0, F_1, \dots, F_{n-d+1})$ by letting $F_i$ denote the number of unordered tuples $\rho \in H^{d-1}$ with $|\lside{\rho} - \rside{\rho}| = i$.

Unfortunately, already for $d=3$ it turns out that this frequency vector does not always determine the number of $w$-embracing tetrahedra, i.e., the number of subsets of $H$ of size $d+1$ whose respective convex hulls contain~$w$.
An example is given in Appendix~\ref{sec:counterexamples2}.

\subsection{Origin-Embracing Sets}

Generalizing the approach for counting embracing $k$-sets from Section~\ref{sec:embrfreqvec} to higher dimension also fails already in 3-space.
Indeed, consider the set of non-embracing tetrahedra for a set $H \subseteq \R^3$.
Observe now that some of these tetrahedra have three edges that form a tangent plane through $w$, whereas others have four such edges.
In fact, the example given in Appendix~\ref{sec:counterexamples2} shows that we cannot hope for a formula with a similar structure as in equation (\ref{eq:embracing_plane}).

Instead, let $H \subseteq \R^3$ be of size $n$ and consider a partition $B \stackrel{.}{\cup} W = H$ defined by a plane~$\phi$ through $w = \mathbf{0}$ that is disjoint from~$H$.
Then, the set of non-embracing $k$-sets consists of those completely in $B$ and $W$, respectively, and those intersected by~$\phi$.
For the latter, consider the intersection of $\conv{A}$ of such a set $A = \{h_1,h_2,\dots\}$ with~$\phi$.
If we restrict our attention to the plane $\phi$, then there is again a unique tangent point $t$ at the intersection of an edge $h_1h_2$ with $\phi$ such that $\conv{A} \cap \phi$ is on the left side of the directed line $wt$.
Hence, with $\lside{h_1h_2}$ being the number of points ``left'' of the plane spanned by $w$, $h_1$, and $h_2$, we get that the number of embracing $k$-sets in 3-space is
\begin{equation}
\embr{k} = \binom{n}{k} - \binom{|B|}{k} - \binom{|W|}{k} - \sum_{h_1,h_2 \in B \times W} \binom{\lside{h_1h_2}}{k-2} \enspace .
\end{equation}

We can generalize this approach in the following way.
\begin{lemma}\label{lem:recurrence_d_space}
Let $H$ be a set of $n$ points in $\R^d$, with $H \cup \{\mathbf{0}\}$ in general position, and let $\psi$ be a generic 2-flat containing the origin.
Let $\mathrm{proj}: \R^d \rightarrow \R^{d-2}$ be a surjective projection that maps all of $\psi$ to $\mathbf{0} \in \R^{d-2}$.
Then, the number of embracing $k$-sets in $H$ is
\begin{eqnarray*}
& &\embr{k}(\mathrm{proj}(H)) - \frac{1}{2} \sum_{\substack{\rho \in \binom{H}{d-1}\\ \conv{\rho} \cap \psi \neq \emptyset}} \left(\binom{\lside{\rho}}{k-d+1} + \binom{\rside{\rho}}{k-d+1}\right) \enspace .
\end{eqnarray*}
\end{lemma}
\begin{proof}
Clearly, any embracing $k$-set is also an embracing $k$-set in the projection, so we only have to subtract the number of non-embracing $k$-sets which happen to be embracing in the projection.
Let $A$ be such a set.
Since $\origin \in \mathrm{proj}(\conv{A})$, we have $\conv{A} \cap \psi \neq \emptyset$.
In the 2-dimensional subspace defined by $\psi$, there is a unique point $t$ on the boundary of $\conv{A} \cap \psi$ such that $\conv{A} \cap \psi$ is in the left closed halfplane defined by $\origin t$. 
Since $\psi$ is generic, $t$ is the intersection of $\psi$ with a $(d-2)$-simplex defined by a tuple $\rho$ of $d-1$ points of~$A$, and the oriented $(d-1)$-simplex $\origin \rho$ has all points of $A \setminus \rho$ either on its positive or negative side.
In the sum, we are thus counting each such non-embracing $k$-set twice (for the left and the right tangent), and the lemma follows.
\end{proof}

With the previous lemma at hand, it is now a simple task to give a proof of our main computational result.

\thmsimpldepth
\begin{proof}
The proof of Lemma~\ref{lem:recurrence_d_space} is constructive apart from the choice of the plane~$\psi$, which can be done arbitrarily using the techniques in~\cite{sos}.
Whether a $(d-2)$-simplex $\conv{\rho}$ intersects $\psi$ can be decided in polynomial time in $d$ (e.g., by testing if the set $\mathrm{proj}(\rho)$ in $\R^{d-2}$ contains the origin $\origin = \mathrm{proj}(\psi)$ in its convex hull).
In order to compute the values $\lside{\rho}$ (and similarly $\rside{\rho}$) for the $(d-1)$-tuples~$\rho$,
we first consider $H$ as a set of $n$ points in the $(d-1)$-dimensional projective plane.
We then compute the dual hyperplane arrangement in $O(n^{d-1})$ time~\cite{arrangements}, which allows the extraction of the values~$\lside{\rho}$ within the same time bounds, as also discussed in~\cite{arrangements}.

After $O(n^{d-1})$ time, we can thus produce a vector that indicates, for each $i$, the number of unordered $(d-1)$-tuples $\rho$ whose convex hull intersects $\psi$ and for which $\lside{\rho} = i$ holds (and hence $\rside{\rho} = n-d+1-i$).
For evaluating the sum in Lemma~\ref{lem:recurrence_d_space}, if we use this vector we have to add up only $O(n)$ binomial coefficients.

Finally, the term $\embr{k}(\mathrm{proj}(H))$ can be evaluated recursively.
Since the dimension decreases by two in each step, the number of recursive calls is linear in the parameter $d$.
\end{proof}

\subsection{Polytopes with Few Vertices}
\label{sec:poly_few_vertices}
Through the so-called Gale transform~\cite{WagnerW01,Welzl01,Ziegler95}, origin-embracing triangles are in correspondence to facets ($(d-1)$-faces) of simplicial $d$-polytopes with at most $n = d+3$ vertices.
More generally, subsets of size $k$ in the so-called Gale dual that contain the origin in their convex hull correspond to $(n-k-1)$-faces of the polytope.
Therefore, some of our results (number of frequency vectors, number of order types, computation of the number of embracing triangles, etc.) have a connection to such simplicial $d$-dimensional polytopes with at most $d+3$ vertices, and thus to known results in that context.

\subparagraph*{Gale duality.}
For $n > d$, we call a matrix $A \in \R^{n\times d}$ \emph{legal} if $A$ has full rank $d$ and $A^\top \mathbf{1}_n = \mathbf{0}_d$.
Let $S_A = (p_1,p_2\commadots p_n)$ be the sequence of points in $\R^d$ with the coordinates of $p_i$ obtained from the $i$-th row of $A$.
Legal thus means that $S_A$ is not contained in a hyperplane and that the origin is the centroid of $S_A$. 
For legal matrices $A \in \R^{n\times d}$ and $B \in \R^{n\times n-d-1}$, we call $B$ an \emph{orthogonal dual} of $A$, in symbols $A \bot B$, if $A^\top B = \mathbf{0}$.
$S_B$ is then called a \emph{Gale dual (Gale transform, Gale diagram) of }$S_A$.%
\footnote{Following~\cite{Welzl01}, we add the requirement that the origin is the centroid, in contrast to
, e.g.,~\cite[Chapter~5.6]{lectures_discrete_geometry}.}
In other words, if $A \bot B$ then all column vectors of $B$ are orthogonal to all column vectors of $A$; together with the legal condition, this means that the column vectors of $B$ span the space of all vectors orthogonal to the columns of $A$ and to $\mathbf{1}_n$, i.e., they form a basis of the null space of $(A,\mathbf{1}_n)^\top$, where $(A,\mathbf{1}_n)$ is the matrix $A$ with an extra column of $1$'s.

The following proposition allows to make a connection between the facets of the polytope defined by the points $S_A$ and its Gale dual $S_B$.

\begin{prop}[{\cite[p.~111]{lectures_discrete_geometry}}]\label{prop:embracing_dual}
Let $A \bot B$ with $S_A = (p_1,p_2\commadots p_n)$ and $S_B = (p_1^*, p_2^*\commadots p_n^*)$. 
For every $I \subseteq [n]$, the set $\{p_i \,|\, i \in I\}$ is contained in a facet of the polytope $\conv{S_A}$ if and only if the set $\{p_i^* \,|\, i \not\in I\}$ is embracing.
\end{prop}


Moreover, given $S_A$ we can show that computing a Gale dual $S_B$ can be done essentially as fast as matrix multiplication.

\begin{prop}\label{prop:computedual}
Given a legal matrix $A \in \R^{n \times d}$, an orthogonal dual can be computed in time $O(n^\omega)$, where $\omega$ is the exponent for matrix multiplication over~$\R$.
\end{prop}
\begin{proof}
Note that $(A,\vecone_n) \in \R^{n \times d+1}$ also has full rank $d+1$ because the extra column $\vecone_n$ is orthogonal to all columns in $A$.
Recall that, therefore, there exists a factorization $(A,\vecone_n)^\top = LUP$ where $L \in \R^{d+1 \times d+1}$ is lower triangular, $U = (U_1,U_2) \in \R^{d+1 \times d+1} \times \R^{d+1 \times n-d-1}$ is upper triangular (in particular, this means that all entries in the diagonals of $L$ and $U_1$ are non-zero and that these matrices are invertible), and $P \in \R^{n \times n}$ is a permutation matrix \cite[Theorem 16.4]{burgisser}.
Also recall that this factorization can be computed in time $O(n^{\omega})$ \cite[Theorem 16.5]{burgisser}.

Given the factorization it is now easy to compute an orthogonal dual $B \in \R^{n \times n-d-1}$. Indeed, write the rows in $B$ such that
\begin{equation*}
  PB = \left(\begin{array}{c}B_1 \\ I\end{array}\right)\enspace ,
\end{equation*}
with unknown $B_1 \in \R^{d+1 \times n-d-1}$ and identity matrix $I \in \R^{n-d-1 \times n-d-1}$. Then, the equality
\begin{equation*}
  \veczero = (A,\vecone_n)^\top B = LUPB
\end{equation*}
is satisfied if and only if
\begin{equation*}
  \veczero = UPB = U_1B_1 + U_2 \Leftrightarrow B_1 = -U_1^{-1}U_2 \enspace .
\end{equation*}
Note that $B$ must have full rank since the columns are clearly linearly independent.
Moreover, as a last step, in order to obtain $B_1$ we first need the inverse $U_1^{-1}$, which can be computed in time $O(d^\omega)$ \cite[Proposition 16.6]{burgisser}.
%
\end{proof}

Another crucial property of the Gale dual is that the points in $S_A$ are in general position if and only if the points in $S_B$ are linearly independent~\cite[p.~111]{lectures_discrete_geometry}.

\subparagraph*{F-vector.}
For a $d$-dimensional polytope $\mathcal{P}$, the \emph{$f$-vector of $\mathcal{P}$} is defined as the sequence $f(\mathcal{P}) = (f_{-1},f_0\commadots f_{d-1})$, where $f_i(\mathcal{P})$ is the number of $i$-dimensional faces ($i$-faces) of $\mathcal{P}$ (the empty face is the unique $(-1)$-face, $0$-faces are vertices, 1-faces are edges, \dots, $(d-1)$-faces are facets).
Thus, if $S = \{p_1,p_2\commadots p_n\}$ is a set of $n$ points in general position, $\mathcal{P} = \conv{S}$ is the corresponding simplicial polytope, and $Q$ is the corresponding Gale dual $\{p_1^*, p_2^*\commadots p_n^*\}$, then $f_i(\mathcal{P})$ is equal to the number of embracing ${(n-i-1)}$-sets in $Q$.
Note that linear independence of the points in $Q$ alone does not assure general position, but slightly perturbing each point on a straight line through the origin does.
Computing the $f$-vector can thus be done by first computing the Gale dual (Proposition~\ref{prop:computedual}), by making use of the correspondence with embracing sets (Proposition~\ref{prop:embracing_dual}), and then by counting these embracing sets (Theorem~\ref{thm:simpldepth}).

\thmcalcfacet

Note that, in spite of the obtained running time, the asymptotic number of facets may be as large as $n^{k}$.
Moreover, a generalization of Theorem~\ref{thm:calc_facet} to sets not necessarily in general position is possible for $k=3$, as explained in Appendix~\ref{sec:multisets}.

\subparagraph{Polytopes with at most $\mathbf{d+3}$ vertices.}
We finally draw the connection between wheel sets of size $n+1$ and the combinatorial structure of simplicial $d$-polytopes with at most $n = d+3$ vertices.

Let $\mathcal{P}$ be the polytope defined as the convex hull of a set $S$ of $n=d+3$ points in general position in $\R^d$.
In other words, let $\mathcal{P}$ be any simplicial $d$-polytope with at most $d+3$ vertices.
Furthermore, let $H$ be the Gale dual of $S$.
Note that $H$ lives in $n-d-1=2$ dimensions, i.e., in the plane.
After possibly rescaling the points in $H$ and adding an extra point $w=\mathbf{0}$ at the origin, we can associate a wheel set $P = H \cup \{w\}$ with the polytope $\mathcal{P}$.
The order type of that wheel set $P$ is uniquely determined by the combinatorial type of the polytope $\mathcal{P}$; indeed, the combinatorial type of $\mathcal{P}$ determines by Propostion~\ref{prop:embracing_dual} the family of $w$-embracing triangles in $H$, which in turn determines by Lemma~\ref{lem:trgs_det_ordtype} the order type of $P$.
In the other direction, each order type of a wheel set gives rise to a unique combinatorial type of a $d$-polytope, since it clearly determines its own family of embracing triangles.
After subtracting one (for the set in convex position, which does not correspond to an actual polytope), it is therefore no coincidence that the number in Theorem~\ref{thm:ordtype} is the same as the one obtained by Perles for the number of simplicial $d$-polytopes with at most $d+3$ vertices \cite[Chapter~6.3]{gruenbaum_book}.

Similarly, given the $f$-vector $f(\mathcal{P})$ of such a polytope, we see by Proposition~\ref{prop:embracing_dual} and Lemma~\ref{lem:embracing_implies_frequency} that it corresponds uniquely to the frequency vector $F(P)$ of a wheel set.
Therefore, the number of $f$-vectors of simplicial $d$-polytopes with at most $d+3$ vertices, as obtained by Linusson~\cite{Linusson99}, equals the number of frequency vectors of wheel sets, as given by Theorem~\ref{thm:freqvec} (again, after subtracting one).
Via the Gale dual, we thus obtain a direct proof for the number of $f$-vectors, as desired by Linusson.
Doing the same for $d+4$ vertices remains an open problem, however.


\bibliographystyle{abbrv}
\bibliography{bibliography}


\appendix

%
%
%
%
%
%
%

\section{Example: Convex Partitions}\label{sec:convex_partitions}
A crossing-free \emph{convex partition} of $P$ is a partition of $P$ such that the convex hulls of the individual parts are pairwise disjoint.
These objects have a natural representation as crossing-free geometric graphs on $P$, simply by taking all edges that lie on the boundary of the convex hull of each part.
Even though there is no obvious choice for $\mathcal{G}$ such that $\nb{\mathcal{G}}{P}$ is equal to the number of crossing-free convex partitions on a conowheel set $P$, it is possible to apply the machinery developed in Section~\ref{sec:ggraph}.
All that is required is a specialized version of Lemma~\ref{lem:mutation}, which shows that $\delta_{i,j}$ is well-defined.
Alternatively, Theorem~\ref{thm:ggraph} could also be generalized to a setting where $\mathcal{G}$ is a family of hypergraphs, but we will not explore that further.

We will make use of the fact that the number of crossing-free convex partitions on $\Pconv$ (a set of $n+1$ points in convex position) is $C_{n+1}$, the $(n+1)$-th Catalan number~\cite{Becker48}.

Let now $\delta_{i,j}$ be the increment in the number of crossing-free convex partitions when going from $P$ to $P'$ as in Figure~\ref{fig:mutation}.
Note that any partitions where $h_1$ and $h_2$ belong to different parts are not affected by the mutation.
The same holds for partitions where $h_1$, $h_2$ and $w$ all belong to the same part.
Hence, when counting partitions on $P$ we may restrict our attention to those cases where $h_1$ and $h_2$ belong to the same part, but that part does not contain $w$.
In any such case, the part that contains $h_1$ and $h_2$ cannot contain any of the $i$ points of $H$ on the left. 
Therefore, the points on the left (without $h_1$ and $h_2$ but including $w$) form a set of $i+1$ points in convex position, giving $C_{i+1}$ possibilities to build a crossing-free convex partition.
The points on the right hand side (with $h_1$ and $h_2$ contracted to a single point) form a set of $j+1$ points in convex position, giving $C_{j+1}$ possibilities.
The set $P'$ can be handled in the same way, and hence we obtain $\delta_{i,j} = C_{i+1}C_{j+1} - C_{i+1}C_{j+1} = 0$.

We conclude that the number of crossing-free convex partitions on any conowheel set of size $n+1$ is $C_{n+1}$.

\section{Frequency Vectors do not determine Triangulations}\label{sec:counterexamples1}
As demonstrated in Section~\ref{sec:crossing_number}, it can sometimes be useful to assign a frequency vector to every point in an arbitrary planar point set $\tilde{P}$.
However, while this set of frequency vectors determines for example the crossing number of $\tilde{P}$, we show here that it does not determine the number of certain crossing-free geometric graphs.

If we allow more than one non-extreme point, there are plenty of examples where the sets of frequency vectors are the same but the numbers of triangulations, say, are different, one of which can be seen in \fig{fig:same_freq_vec}.
Intuitively, in that example one of the non-extreme points $w$ moves over a halving segment $h_1h_2$, which maintains the frequency vector of $w$, and without encountering any other collinearities, which maintains the frequency vectors of all the other points.
At the same time, however, moving $w$ changes the number of triangulations.
For the illustrated example, the latter can be seen easily by noting that we only have to consider triangulations that use the edge $\{h_1,h_2\}$ if we want to compute the difference, similar to what we did in the proof of Lemma~\ref{lem:mutation}.
Under this contraint, $\tilde{P}$ has $3 \cdot 1 = 3$ triangulations (we are simply multiplying the respective numbers of triangulations to the left and right of the segment $h_1h_2$), whereas $\tilde{P}'$ has $2 \cdot 2 = 4$ triangulations.
We conclude that $\tilde{P}$ has one fewer triangulation than $\tilde{P}'$.

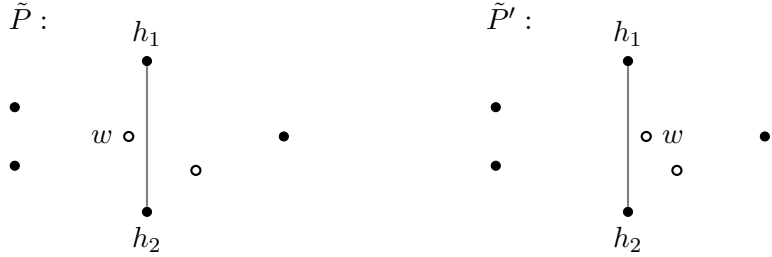
\begin{figure}
\centering
     \begin{tikzpicture}[xscale=1.2]
       \begin{scope}[xshift=-75]
        \node[Hpt,label=above:$h_1$] (p) at (90:1) {};
         \node[Hpt,label=below:$h_2$] (q) at (270:1) {};
         \node[Wpt,label=left:$w$] (w) at (-0.2,0) {};
         \node[Wpt] (w2) at (320:0.7) {};
         \node[Hpt] at (165:1.5) (l1) {};
         \node[Hpt] at (195:1.5) (l3) {};
         \node[Hpt] at (0:1.5) (r1) {};
         \draw[sptline3] (p) -- (q);
         
         \node at (-1.3,1.6) {$\tilde{P}:$};
       \end{scope}
       
       \begin{scope}[xshift=75]
         \node[Hpt,label=above:$h_1$] (p) at (90:1) {};
         \node[Hpt,label=below:$h_2$] (q) at (270:1) {};
         \node[Wpt,label=right:$w$] (w) at (+0.2,0) {};
         \node[Wpt] (w2) at (320:0.7) {};
         \node[Hpt] at (165:1.5) (l1) {};
         \node[Hpt] at (195:1.5) (l3) {};
         \node[Hpt] at (0:1.5) (r1) {};
         \draw[sptline3] (p) -- (q);
         
         \node at (-1.3,1.6) {$\tilde{P}':$};
       \end{scope}
     \end{tikzpicture}
\caption{Two point sets $\tilde{P}$ and $\tilde{P}'$ (each with two non-extreme points) that have the same set of frequency vectors, but a different number of triangulations.}
\label{fig:same_freq_vec}
\end{figure}

\section{Frequency Vector does not determine Embracing Tetrahedra}
\label{sec:counterexamples2}
As discussed in Section~\ref{sec:higher_dimensions}, the concepts of conowheel sets and frequency vectors can be generalized to higher dimensions.
However, already in $\R^3$ there are conowheel sets with the same such 3-dimensional frequency vector but with a different number of $w$-embracing tetrahedra.
In what follows, we explain the example illustrated in \fig{fig:same_freq_vec_diff_tetrahedra}.

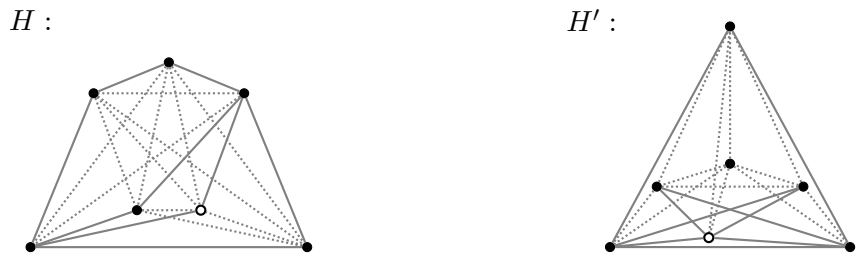
\begin{figure}[b]
\centering
\begin{tikzpicture}[scale=1.4]
       \begin{scope}[xshift=-75]
         \node[Hpt] (b1) at (45:1) {};
         \node[Hpt] (b2) at (90:1) {};
         \node[Hpt] (b3) at (135:1) {};
         \node[Hpt] (b4) at (210:1.5) {};
         \node[Hpt] (b5) at (330:1.5) {};
         \node[Hpt] (b6) at (-0.3,-0.4) {};
         \node[Wpt] (w) at (0.3,-0.4) {};

         \draw[sptline3] (b1) -- (b2);
         \draw[sptline3] (b2) -- (b3);
         \draw[sptline3] (b3) -- (b4);
         \draw[sptline3] (b4) -- (b5);
         \draw[sptline3] (b5) -- (b1);
         \draw[sptline3] (b1) -- (b6);
         \draw[sptline3] (b6) -- (b4);
         \draw[sptline3] (b4) -- (w);
         \draw[sptline3] (w) -- (b1);

         \draw[sptline4] (b1) -- (b3);
         \draw[sptline4] (b1) -- (b4);
         \draw[sptline4] (b2) -- (b4);
         \draw[sptline4] (b2) -- (b5);
         \draw[sptline4] (b2) -- (b6);
         \draw[sptline4] (b2) -- (w);
         \draw[sptline4] (b3) -- (b5);
         \draw[sptline4] (b3) -- (b6);
         \draw[sptline4] (b3) -- (w);
         \draw[sptline4] (b5) -- (b6);
         \draw[sptline4] (b5) -- (w);
         \draw[sptline4] (b6) -- (w);
         
         \node at (-1.3,1.4) {$H:$};
       \end{scope}
       
       \begin{scope}[xshift=75]
         \begin{scope}[yshift=-33]
         \node[Hpt] (b1) at (20:1.2) {};
         \node[Hpt] (b2) at (55:1.2) {};
         \node[Hpt] (b3) at (90:1.2) {};
         \node[Hpt] (b4) at (125:1.2) {};
         \node[Hpt] (b5) at (160:1.2) {};
         \node[Hpt] (b6) at (90:2.5) {};
         \node[Wpt] (w) at (-0.2,0.5) {};
         \end{scope}

         \draw[sptline3] (b1) -- (b6);
         \draw[sptline3] (b6) -- (b5);
         \draw[sptline3] (b5) -- (b1);
         \draw[sptline3] (b1) -- (b4);
         \draw[sptline3] (b4) -- (w);
         \draw[sptline3] (w) -- (b1);
         \draw[sptline3] (b2) -- (b5);
         \draw[sptline3] (b5) -- (w);
         \draw[sptline3] (w) -- (b2);

         \draw[sptline4] (b1) -- (b2);
         \draw[sptline4] (b1) -- (b3);
         \draw[sptline4] (b2) -- (b3);
         \draw[sptline4] (b2) -- (b4);
         \draw[sptline4] (b2) -- (b6);
         \draw[sptline4] (b3) -- (b4);
         \draw[sptline4] (b3) -- (b5);
         \draw[sptline4] (b3) -- (b6);
         \draw[sptline4] (b3) -- (w);
         \draw[sptline4] (b4) -- (b5);
         \draw[sptline4] (b4) -- (b6);
         \draw[sptline4] (b6) -- (w);
         
         \node at (-1.3,1.4) {$H':$};
       \end{scope}
\end{tikzpicture}
\caption{Point sets $H$ and $H'$ in $\R^3$, projected onto the drawing plane $z=1$ by lines through the origin.
Points with negative $z$-coordinate are white and the others are black.
Combined with an extra point $w=\mathbf{0}$ at the origin (which is not depicted in the figure), this defines two conowheel sets.
Both resulting sets have the same 3-dimensional frequency vector, but $H$ has six origin-embracing tetrahedra, whereas $H'$ has only four.
}
\label{fig:same_freq_vec_diff_tetrahedra}
\end{figure}

In that figure we use a representation discussed by Stolfi~\cite{stolfi} in the context of ``oriented projective geometry'':
A point $p = (x, y, z)$ in $\R^3$ is projected to the point $p' = (x/z, y/z)$ in $\R^2$.
In more intuitive terms, we project $p$ onto the drawing plane $z=1$ along a line that goes through the origin.
In the figure, we further distinguish points with negative $z$-coordinate (white) and positive $z$-coordinate (black).

The figure depicts two distinct point sets $H$ and $H'$, which, when combined with an extra point $w=\mathbf{0}$ at the origin, yield two conowheel sets $P$ and $P'$ in $\R^3$.
The employed projection allows us to identify origin-embracing tetrahedra; either a segment between two white points that crosses a segment between two black points, or a black point inside a triangle spanned by three white points, or a white point inside a triangle spanned by three black points.
Given that there is only one white point, only the last case is relevant; and a careful counting shows that $H$ has six origin-embracing tetrahedra, whereas $H'$ has only four.

Furthermore, the projection allows for counting points on one side of the plane spanned by two points $h_1$, $h_2$ of $H$, say, and the origin $w$; we simply take the number of black points on one side of the line through $h_1$ and $h_2$ in the projection, and then add the number of white points on the other side.
In the figure we thus distinguish pairs of points where this sum is one (solid lines) and where it is two (dotted lines).
Consequently, a careful counting again shows that the two sets have the same 3-dimensional frequency vector $F(P) = F(P') = (0,12,0,9,0,0)$.

\section{Minimal Embracing Multisets}\label{sec:multisets}
We give a short account on how our approach for Theorem~\ref{thm:calc_facet} in Section~\ref{sec:poly_few_vertices} can be modified for point sets not in general position in the case $k=3$.
That is, we need to understand how to count minimal embracing sets in $\R^2$, given that there might be multiplicities in the point set.

\newcommand{\embrmin}{\mathsf{embr}_{\mathsf{min}}}
Let $H$ be a multiset of $n$ points in the plane on the unit circle centered at the origin, which again takes the role of the extra point $w$.
We again let $\lside{h}$ and $\rside{h}$ be the numbers of points to the left and right, respectively, of the line $wh$.
In addition, we let $\multiplicity{h}$ denote the multiplicity of $h$ in $H$, and we let $\opposite{h}$ denote the multiplicity of the point $-h$ in $H$, i.e., the number of points on the line $wh$ but on the ``opposite'' side of $h$.
For any $h \in H$ we thus have $n = \lside{h} + \rside{h} + \multiplicity{h} + \opposite{h}$.
Finally, we denote by $\tilde H$ the underlying set of points contained in $H$, i.e., $\tilde H$ is $H$ with all multiplicites removed.

A subset $A \subseteq H$ of points is called \emph{minimal $w$-embracing} if it is $w$-embracing, but no proper subset of~$A$ is $w$-embracing.
Note that any such $A$ contains either three distinct points with $w$ in the interior of $\conv{A}$, or two distinct points with $w$ in the relative interior of the connecting segment.
The case where $A$ has size four or larger cannot occur since it is not minimal, and the case where $A$ has size one cannot occur because $w$ is not contained in $H$.

The minimal $w$-embracing subsets of size two are easy to count: 
\begin{equation}
\label{eq:MinEmbrSeg}
\frac{1}{2}   \sum_{h \in \tilde{H}} \multiplicity{h} \opposite{h} = \frac{1}{2} \sum_{h \in H} \opposite{h} \enspace . 
\end{equation}

To understand minimal $w$-embracing subsets of size three, we first compute the number $\Delta(H)$ of \emph{proper} triangles in $H$, i.e., subsets of three points that do not lie on a common line:\footnote{We use that for $n = a_1 + a_2 + \dots + a_m$, we have 
$\sum_{\{i,j,k\} \in {\binom{[m]}{3}}} a_i a_j a_k = \frac{1}{6} \sum_{i=1}^m a_i (n-a_i) (n - 2a_i)$.}
\begin{equation}\label{eq:proper_triangles}
\begin{aligned}
\Delta(H) = \sum_{\{h,p,q\} \in {\binom{\tilde{H}}{3}}} \multiplicity{h} \multiplicity{p} \multiplicity{q} &= \frac{1}{6} \sum_{h \in \tilde{H}} \multiplicity{h}(n -\multiplicity{h})(n-2\multiplicity{h})\\
&= \frac{1}{6}  \sum_{h \in H}   (n -\multiplicity{h})(n-2\multiplicity{h}) \enspace.
\end{aligned}
\end{equation}

Now call a pair  $(h,\{p,q\})$ of three distinct points in $H$ \emph{angle-embracing} if $w$ lies in the interior of the infinite convex cone that has apex $h$ and that is spanned by vectors $p-h$ and $q-h$.
Note that if $A = \{h,p,q\}$ forms a proper triangle, then this gives rise to three angle-embracing pairs if and only if $A$ is minimal $w$-embracing, and it gives rise to exactly one angle-embracing pair, otherwise.
If $x$ denotes the number of minimal $w$-embracing triangles, and $y$ the number of proper triangles that are not minimal $w$-embracing (which includes, in particular, all proper triangles that are not $w$-embracing at all), we therefore have
\begin{align*}
x + y &= \Delta(H)\\
3x + y &= \sum_{h \in \tilde{H}} \multiplicity{h} \lside{h} \rside{h} \enspace .
\end{align*}
We get the following for $x$, the number of minimal $w$-embracing sets of size three:
\begin{equation}
\label{eq:MinEmbrTrg}
\frac{1}{2} \sum_{h \in \tilde{H}} \multiplicity{h} \lside{h} \rside{h} - \frac{1}{2} \Delta(H) = \frac{1}{2} \sum_{h \in {H}} \lside{h} \rside{h} - \frac{1}{2} \Delta(H) \enspace .
\end{equation}
Summing up equations (\ref{eq:MinEmbrSeg}) and (\ref{eq:MinEmbrTrg}), and combining them with equation (\ref{eq:proper_triangles}), finally gives the following formula for the number of minimal $w$-embracing subsets in $H$:
\begin{equation}
\frac{1}{2} \sum_{h \in H}  \left(\lside{h} \rside{h} + \opposite{h} -\frac{1}{6}  (n -\multiplicity{h})(n-2\multiplicity{h})\right) \enspace .
\end{equation}

\end{document}